\documentclass[10pt]{article}

\usepackage{amsmath,amssymb}
\usepackage{booktabs}
\usepackage{geometry}
\usepackage{float}
\usepackage{authblk}
\usepackage{amsthm}
\usepackage{xcolor}
\usepackage{amscd}
\usepackage{url}
\usepackage{cite}
\usepackage[mathlines]{lineno}
\providecommand{\doi}[1]{\url{https://doi.org/#1}}
\newtheorem{theorem}{Theorem}

\newtheorem{rem}{Remark}
\newtheorem{corollary}{Corollary}
\newtheorem{proposition}{Proposition}

\title{Temporal invariance of evolutionary entropy 
	\\ and comparative demographic analysis}

\author[1]{Jorge Buescu}
\author[2]{Henrique M. Oliveira\thanks{Corresponding author: \texttt{henrique.m.oliveira@tecnico.ulisboa.pt}.}}
\affil[1]{\small Departamento de Matem\'atica, Faculdade de Ci\^encias, Universidade de Lisboa,  Portugal; Centro de Estudos Matem\'aticos, Universidade de Lisboa, Portugal}
\affil[2]{\small Departamento de Matem\'atica, Instituto Superior T\'ecnico, Universidade de Lisboa, Portugal; Centro de An\'alise Matem\'atica, Geometria e Sistemas Din\^amicos, Universidade de Lisboa, Portugal}

\date{\today}

\begin{document}

\maketitle
\begin{abstract}
Does evolutionary entropy depend on the width of the age classes used to describe a population? We show that the apparent dependence on temporal resolution is a consequence of the loss of information about reproductive timing within discrete classes.
When this information is retained,
corresponding continuous and discrete representations yield the same
cycle entropy, generation time, evolutionary entropy, and reproductive potential in a
common unit of time.
 When within-class timing is unknown, a uniform local
 reconstruction yields a
logarithmic correction.
  For constant post-maturity hazard and fertility, continuous exponential and discrete geometric descriptions correspond exactly at every time scale. Thus, for a fixed \textcolor{black}{reference duration}, evolutionary entropy is a
property of the reproductive process, not of its temporal partition.
These results provide a mathematical basis for comparing demographic models across different temporal scales by identifying the information required to separate differences in reproductive organisation from effects of the chosen age-class representation.
\end{abstract}


\section{Introduction}
\label{sec:introduction}
Age is a continuous biological variable, yet demographic data and population models are often represented in discrete age classes. The age-structured equations of McKendrick \cite{McKendrick1926} describe population dynamics in continuous age, whereas Leslie's matrix formulation \cite{Leslie1945} provides the canonical discrete representation. Although the two formulations are mathematically closely related, they raise a fundamental question: which properties are intrinsic to the demographic process, and which are artefacts of the temporal representation?
Characteristic-based and finite-difference approximations of continuous-age
models have been studied with regard to stability, convergence, and numerical
accuracy \cite{AbiaAnguloLopezMarcos2005,PelovskaIannelli2006}. Here the question is instead which
demographic functionals can be recovered exactly from a temporal partition. The relation between McKendrick and Leslie descriptions is also studied in \cite{DilaoLakmeche2006}.

Throughout this article, survivorship, fertility, and hazard are defined
on an underlying continuous-age demographic process; changes of temporal
scale and age-class discretisations are treated as alternative
representations of that same process.
This distinction is especially important for evolutionary entropy. Introduced by Demetrius \cite{Demetrius1974,Demetrius1975}, it characterises the temporal distribution of reproductive contributions in age-structured populations. Evolutionary entropy in discrete models admits an ergodic interpretation in terms of the proliferation of typical genealogies, from which the concepts of generation time and reproductive potential arise naturally. 
These concepts were subsequently developed into a broader theory of structured population dynamics and natural selection; see, in particular, Demetrius and Legendre \cite{DemetriusLegendre2013}.
Related sensitivities of growth and evolutionary entropy are studied in \cite{BuescuOliveiraSousa2023}.

In discrete age-structured models, normalised reproductive
contributions form a probability distribution. Formally, the cycle entropy is
a Shannon entropy \cite{Shannon1948}, but its interpretation here is demographic:
the probabilities describe reproductive contributions at different ages, and the
entropy characterises their distribution across the life cycle
\cite{Demetrius1974,Demetrius1975}.
Shannon entropy depends on the partition of the underlying probability distribution: using finer age classes generally changes both the class probabilities and the resulting entropy.  This dependence on resolution is classical in information theory and dates back at least to
 R\'enyi
\cite{Renyi1959}.
The corresponding relation between quantised and differential entropy is
reviewed by Gray and Neuhoff \cite{GrayNeuhoff1998}. In demographic applications,
Demetrius, Legendre, and \textcolor{black}{Harremo{\"e}s} \cite{DemetriusLegendreHarremoes2009} 
explicitly distinguish continuous entropy calculations from their discrete
approximations.

This directly impacts the  demographic problem. 
Age classes are modelling constructs, not biological entities.
The same population may be represented by five-year, annual,  monthly, or even finer age classes. If these descriptions represent the same underlying survival and reproductive process, evolutionary entropy should define the same demographic quantity when expressed in a common physical unit of time.

Evolutionary entropy  $H$ is defined as
\begin{equation*}
	 H = \frac{S}{T},
	 \end{equation*}
where $S$ is the cycle entropy and $T$ is the generation time. Since $S$ is dimensionless,  $H$ has units of inverse time $[T]^{-1}$ and is therefore a rate. We refer to it simply as evolutionary entropy throughout, keeping this dimensional character in mind.

The decomposition of population growth into evolutionary entropy and
reproductive potential is due to Demetrius \cite{Demetrius1974,Demetrius1975}.
With the sign convention used here, it reads \(r=H+\Phi\)
\cite{DemetriusLegendre2013}, where \(r\) is the intrinsic growth rate and
\(\Phi\) is reproductive potential.
Section~\ref{sec:continuous_demography} defines the two terms using the same
temporal reference. The results for \(\Phi\) are stated alongside those for
entropy; their proofs are collected in Appendix~\ref{app:potential_proofs}.
 
 A simple example illustrates the issue. Consider a human population represented in hypothetical five-year age classes. A generation time $T_5=4$ classes corresponds to a physical generation time of 20 years. If the same population is represented in one-year classes, then $T_1=20$ classes. The numerical values differ because the unit has changed, but the physical generation time is the same: $5T_5=T_1=20$ years. The same principle applies to evolutionary entropy. With five-year classes, $H=S/4$ per five-year unit; with one-year classes, $H=S/20$ per year. Since $\frac{1}{5} \, \frac{S}{4} =\frac{S}{20}$, both representations yield the same evolutionary entropy in a common physical unit.

The difficulty is that, unlike generation time, cycle entropy does not automatically transform with the choice of age classes. Representing the same population with five-year or annual classes generally produces different Shannon entropies because the underlying probability distributions are partitioned differently. Denoting by 
$S_5$ and $S_1$  
the Shannon entropies computed from the corresponding class probabilities,  in general $S_5 \neq S_1$. 
Consequently, the evolutionary entropy $H$ 
 appears to depend on the age-class width, even when generation time has been expressed in a common physical unit. The underlying demographic process has not changed; only its temporal representation has.

A related issue arises for survivorship-based entropy: a discretisation
that preserves its interpretation as an elasticity need not preserve its use
as a measure of mortality shape \cite{deVries2023,Giaimo2024,Bernard2024}. This concerns a
different functional from \(H=S/T\), but illustrates why the property to be
preserved must be specified.

The problem addressed in this paper may therefore be stated as follows:
How must evolutionary entropy be computed so that equivalent temporal descriptions of the same demographic process yield identical values when expressed in the same time unit?

The aim is to provide biologists and mathematical biologists with a rigorous
basis for comparing demographic models expressed at different temporal scales.
The results specify which quantities can be compared directly, which require
within-class information, and when an approximation is needed. This helps
distinguish genuine demographic differences from effects of temporal
representation.

To answer this question, it is useful to distinguish four different but related problems. They correspond to the four possible transitions between continuous and discrete temporal descriptions and are summarised in Table~\ref{tab:intro_questions}.

\begin{table}[ht]
\centering
\small
\caption{The four temporal representation problems for evolutionary
entropy.}
\label{tab:intro_questions}
\vspace{2mm}
\renewcommand{\arraystretch}{1.30}
\begin{tabular}{
p{0.12\textwidth}
p{0.20\textwidth}
p{0.25\textwidth}
p{0.32\textwidth}}
\toprule
\textbf{Type}
&
\textbf{Problem}
&
\textbf{Description}
&
\textbf{Question}
\\
\midrule

Continuous
\(\to\)
continuous
&
Change of temporal unit
&
The same continuous process is expressed in a different unit of time
&
How do continuous demographic quantities transform when years are
replaced by months, or more generally when the temporal unit is
changed?
\\

\addlinespace

Continuous
\(\to\)
discrete
&
Temporal aggregation
&
A continuous reproductive-age distribution is replaced by age-class probabilities
&
What information is lost when continuous reproductive timing is
represented only through age-class probabilities?
\\

\addlinespace

Discrete
\(\to\)
continuous
&
Continuous reconstruction
&
Class probabilities, together with within-class timing information, are used to reconstruct a continuous description
&
Can continuous cycle entropy and generation time be recovered from an
age-class representation, asymptotically or exactly?
\\

\addlinespace

Discrete
\(\to\)
discrete
&
Change of temporal mesh
&
The same demographic process is represented with different class widths.
&
How can equivalent discrete representations with different class widths yield the same evolutionary entropy when expressed in a common unit of time?
\\

\bottomrule
\end{tabular}
\end{table}

The continuous-to-continuous problem is a change of time unit. Expressing age in months rather than years does not alter the underlying population;  ages and demographic rates acquire different numerical values because their units change. We show that generation time and evolutionary entropy  follow the corresponding transformation laws.

The continuous-to-discrete problem is temporal aggregation. Reproductive contributions distributed over a continuous range of ages are aggregated into age classes. The total contribution assigned to each class is preserved, but its temporal distribution within the class is generally lost.

The reverse discrete-to-continuous problem is not straightforward. Class probabilities alone do not determine a unique continuous reproductive-age distribution; additional information about within-class reproductive timing is required. The reconstruction problem is to identify what information is missing and how much is needed to recover cycle entropy, generation time, and evolutionary entropy.

The discrete-to-discrete problem is subtler. A coarse model does not, in general, determine its own refinement, since knowing the reproductive contribution assigned to a five-year age class does not specify how it is distributed among the corresponding annual classes. Thus, passing from a five-year to an annual model is not simply an algebraic operation on the coarse class probabilities: additional temporal information is required if the two models are to represent the exact same demographic process.

This requirement also separates temporal reconstruction from changes in
model construction. Converting life tables into stage matrices can affect
estimated growth and generation time \cite{FujiwaraDiazLopez2017}; in stage
models, census conventions can change the assumed stage-duration distributions
and hence eigenvalue analyses \cite{Okuyama2019}. Such changes need not describe
the same underlying process.

Within the Leslie family, Hinrichsen \cite{Hinrichsen2023} gives an
aggregation that preserves stable growth under suitable weighting. The
preprint of Hinrichsen, Yokomizo, and Salguero-G{\'o}mez
\cite{HinrichsenEtAl2026} also evaluates aggregation errors in generation time
and reproductive-distribution entropy. These reduced-matrix comparisons
complement the exact reconstruction of a fixed reproductive-age density
considered here. Stage aggregation through genealogical decompositions is studied in \cite{BienvenuEtAl2017}; entropy-rate preservation under aggregation of Markov states is treated in \cite{GeigerTemmel2014}.

A related but distinct question concerns reproductive amplitude. Should evolutionary entropy depend on a common multiplicative factor applied to all reproductive contributions? Since entropy is computed from normalised reproductive contributions, at fixed \(r\) it should depend only on their relative  distribution. We examine this homogeneity property after defining the continuous reproductive-age distribution.

The temporal questions in Table~\ref{tab:intro_questions} arose naturally in recent comparative work on Leslie-type demographic models \cite{BuescuElaydiOliveira2026}. That study analysed entropy-maximising reproductive windows and a critical threshold for geometric reproductive tails across 130 animal species. Since all projection matrices used annual age classes, evolutionary entropy and the associated critical equation were computed in an annual discrete representation.

This raises a question of temporal consistency. What would change if the same populations were instead represented by monthly age classes? The underlying biology would be unchanged, but the reproductive probabilities, the geometric parameter describing a reproductive tail, and the ordinary Shannon entropy of the class probabilities would all differ. Could the estimated evolutionary entropy, or even the location of its maximum, change solely because the temporal mesh had been refined?

The core of the paper addresses this issue. Temporal aggregation preserves the total reproductive contribution assigned to each age class but generally discards its within-class temporal distribution. Recovering this information is the key to identifying equivalent representations at different resolutions.

For general reproductive-age distributions, decomposing reproductive contributions into between-class and within-class components leads to an exact reconstruction principle. When the within-class profile is unknown, a natural local approximation yields a logarithmic correction associated with the class width and recovers the continuous entropy in the limit
of vanishing class width.

A stronger result holds for populations with constant post-maturity hazard and fertility. In this case, the normalised reproductive-age distribution is exponential after the age at first reproduction, and its aggregation into equal age classes is geometric. The within-class temporal structure is identical throughout the reproductive period and can therefore be recovered exactly.
For fixed mesh $h>0$, the geometric parameter  \(q_h\) 
plays the role of a finite-step counterpart to the continuous decay rate, and a mesh-invariant parameter $\beta$ relates different resolutions.

The invariance principle can now be stated rigorously: equivalent temporal representations have the same reconstructed cycle entropy and physical generation time. Consequently, they yield the same evolutionary entropy when expressed in a common unit of time.

This distinction between mesh-dependent coordinates and mesh-independent demographic quantities is also important for the entropy-maximising critical point. It explains the difference between the critical thresholds of the continuous and conventional discrete models and why, despite not being mathematically identical, they are nevertheless numerically very close for long-lived populations.

The contribution of this paper is the joint reconstruction of \(S,T,H\) and,
when the same Euler--Lotka rate \(r\) is retained, \(\Phi\), across temporal
resolutions. The classical chain rule and quantisation identities
\cite{Shannon1948,Renyi1959,GrayNeuhoff1998} are used to identify sufficient
within-class summaries, without requiring reconstruction of the entire density.
The resulting comparison distinguishes exact mesh invariance, rate covariance
under a change of units, and compensating changes under a new reference duration.
For the exponential--geometric family, we transport the continuous entropy maximum
exactly and quantify the first-order shift of the corrected discrete optimum in
Theorem~\ref{thm:discrete_critical_threshold_limit}. These questions are distinct
from matrix-parameter sensitivity \cite{Oliveira2026} and the longevity-calibrated
application \cite{BuescuElaydiOliveira2026GreenlandShark}.

The paper is organised around the four questions in Table~\ref{tab:intro_questions}. We begin by introducing the continuous reproductive-age distribution together with the corresponding entropy and generation time, and examine their behaviour under changes of temporal unit. We then study temporal aggregation and identify the information lost in passing from continuous to discrete representations. Reversing the process leads to an exact partition decomposition.
When within-class timing is unavailable, a uniform local reconstruction recovers the continuous cycle entropy in the fine-mesh limit. We then compare equivalent discrete representations at different temporal resolutions. 
 Finally, we show that, for the constant-hazard model, continuous and discrete temporal meshes are related by an exact transformation, allowing a comparison of the corresponding critical entropy thresholds.

\section{Continuous reproductive-age distribution}
\label{sec:continuous_demography}

Let \(a\geq0\) denote age. If $l(a)$ is the probability of survival from birth to age $a$ and $m(a)$ is the age-specific fertility rate, then the net reproductive function is
\begin{equation}
V(a)=l(a)m(a).
\label{eq:net_reproductive_function}
\end{equation}
The intrinsic growth rate $r$ is defined by the continuous Euler--Lotka equation
\begin{equation}
\int_0^\infty e^{-ra}V(a)\,da=1.
\label{eq:euler_lotka_continuous}
\end{equation}

The left-hand side is the Laplace transform of \(V\). Under the 
assumptions that \(V\geq0\) is nonzero and integrable, and that reproductive lifespan is finite, it defines an entire function of exponential
type. Restricted to the real axis, it is completely monotone in the
sense of Bernstein and tends to zero as \(r\to+\infty\). It is strictly decreasing and tends to \(+\infty\) as \(r\to-\infty\). It follows
that \eqref{eq:euler_lotka_continuous} has a unique real solution. For unbounded reproductive ages, existence of a real Euler--Lotka solution is assumed.

For a fixed value of \(r\) with \(0<\int_0^\infty e^{-ra}V(a)\,da<\infty\), the normalised reproductive-age
distribution associated with \(V\) is
\begin{equation}
p_r(a)
=
\frac{
e^{-ra}V(a)
}{
\displaystyle
\int_0^\infty e^{-rs}V(s)\,ds
},
\qquad a\geq0.
\label{eq:continuous_reproductive_density}
\end{equation}
At the Euler--Lotka value of \(r\),
\eqref{eq:continuous_reproductive_density} reduces to
\begin{equation*}
p(a)=e^{-ra}V(a),
\end{equation*}
with
$
\int_0^\infty p(a)\,da=1.
$
The probability density \(p\) describes the temporal distribution of normalised
reproductive contribution. It is the continuous counterpart of the
reproductive probability distribution used in discrete age-structured
models.

The normalisation in
\eqref{eq:continuous_reproductive_density} implies $p_r$ is homogeneous of degree zero in $V$; that is, the  map
\[
V(a)\longmapsto cV(a),
\qquad c>0,
\]
leaves \(p_r\) unchanged. Thus \(p_r\) depends only on the relative
temporal distribution of reproductive contributions.
In particular, this invariance holds under a common multiplicative rescaling of all fertility rates, or more generally of the net reproductive function \(V\).
It applies to the normalised reproductive distribution for fixed
\(r\), and should not be confused with rescaling the demographic
window and solving the Euler--Lotka equation again, which in general
changes the intrinsic growth rate $r$.

The generation time, or mean reproductive age, is the first temporal moment of $p$:
\begin{equation*}
T
=
\int_0^\infty a\,p(a)\,da.
\end{equation*}
A genealogical interpretation of generation time in matrix models is given in \cite{BienvenuLegendre2015}.

The probability density $p$ has dimension of inverse time $[T]^{-1}$, so the argument of the logarithm in the entropy integral requires a choice of time scale.
Fix a positive reference duration and let \(\tau>0\) be its numerical
value in the same time unit as age \(a\). A fixed one-year reference has
\(\tau=1\) when ages are measured in years and \(\tau=12\) when they are
measured in months. The product \(\tau p(a)\) is therefore dimensionless.
Define the continuous cycle entropy and the corresponding rate by
\begingroup
\begin{equation}
S_\tau=-\int_0^\infty p(a)\log\!\bigl(\tau p(a)\bigr)\,da,
\qquad H_\tau=\frac{S_\tau}{T}.
\label{eq:continuous_cycle_entropy}
\end{equation}
\par\endgroup

Unless stated otherwise, the reference duration is one year. We write
\(S=S_\tau\) and \(H=H_\tau\) when this reference is fixed; in expressions
with ages measured in years, \(\tau=1\). Thus \(H\) is expressed in
\(\mathrm{yr}^{-1}\). Throughout, \(0<T<\infty\),
\(\int_0^\infty p(a)|\log(\tau p(a))|\,da<\infty\), and the class entropies
under consideration are finite.

\begingroup
\begin{rem}[Differential and discrete entropy]
\label{rem:differential_entropy_sign}
The integral in \eqref{eq:continuous_cycle_entropy} is a differential
entropy \cite{Shannon1948,DemetriusLegendreHarremoes2009}.
Unlike probabilities, density values need not lie in \([0,1]\), so this
integral need not be non-negative. For example, fix \(a_0\geq0\) and a
window length \(\varepsilon>0\). The uniform density on that window satisfies
\begin{equation*}
p_\varepsilon(a)=
\begin{cases}
1/\varepsilon,&a_0\leq a\leq a_0+\varepsilon,\\
0,&\text{otherwise},
\end{cases}
\qquad
S_\tau[p_\varepsilon]=\log(\varepsilon/\tau).
\end{equation*}
Here \(S_\tau[p_\varepsilon]\) means \eqref{eq:continuous_cycle_entropy}
evaluated at \(p_\varepsilon\); it is negative when \(\varepsilon<\tau\).
By contrast, \(-\sum_j p_j\log p_j\geq0\) for every discrete probability
distribution \((p_j)_j\). The reconstruction results below concern the
reference-dependent differential entropy, not an identification of these
two functionals or a claim that the reconstructed entropy is non-negative.
\end{rem}
\par\endgroup

\begingroup
With the same reference duration, define the reproductive potential
\cite{DemetriusLegendre2013,DemetriusGundlach2014} by
\begin{equation*}
\Phi_\tau=\frac1T\int_0^\infty
p(a)\log\!\bigl(\tau V(a)\bigr)\,da.
\end{equation*}
The integrand is taken to be zero wherever \(p(a)=0\). The integral is finite
under the assumptions above, since \(\log(\tau V(a))=ra+\log(\tau p(a))\)
where \(p(a)>0\). We write \(\Phi=\Phi_\tau\) when the one-year reference
is understood.
\begin{proposition}[Growth decomposition]
\label{pro:potential_decomposition}
For the reproductive density \(p(a)=e^{-ra}V(a)\), where \(r\) satisfies
Euler--Lotka and the finiteness assumptions above hold,
\begin{equation}
r=H_\tau+\Phi_\tau.
\label{eq:entropy_potential_decomposition}
\end{equation}
\end{proposition}
\par\endgroup

This result is proved in Appendix~\ref{app:potential_proofs}, together with all subsequent results on the reproductive potential \(\Phi\), which follow from the fundamental decomposition \(r=H+\Phi\) and the corresponding results for \(H\).

The variables 
\(p\), \(T\), \(S\), \(H\), and \(\Phi\)
define the continuous
formulation adopted in this paper. 
 The next section addresses the
first question in Table~\ref{tab:intro_questions}: how these quantities
transform when the same continuous demographic process is expressed in
a different unit of time.

\section{Change of temporal unit}
\label{sec:reference_scale}

We first consider the continuous-to-continuous problem in
Table~\ref{tab:intro_questions}. A change of temporal unit is simply a
change of coordinates: the demographic process is unchanged, whereas
the age variable and all time-dependent quantities are rescaled.

Consider a linear change of time scale
\begin{equation}
	\widetilde a=ca,
	\qquad c>0,
	\label{eq:S_scale_change}
\end{equation}
where \(a\) and \(\widetilde a\) denote the same physical age expressed
in different temporal units. For example, \(c=12\) when passing from
years to months, and \(c=5\) when passing from five-year units to
years.

Survivorship is a probability and  transforms by
composition with the scaling law \eqref{eq:S_scale_change}:
\begin{equation*}
\widetilde l(\widetilde a)
=
l\!\left(\frac{\widetilde a}{c}\right).
\end{equation*}
Fertility is a rate per unit time and transforms according to
\begin{equation*}
\widetilde m(\widetilde a)
=
\frac{1}{c}
m\!\left(\frac{\widetilde a}{c}\right).
\end{equation*}
Consequently, from \eqref{eq:net_reproductive_function} we obtain
\begin{equation}
\widetilde V(\widetilde a)
=
\frac{1}{c}
V\!\left(\frac{\widetilde a}{c}\right).
\label{eq:decomposition_scale_invariant}
\end{equation}
The intrinsic growth rate is also a rate per unit time, so
$
\widetilde r=\frac{r}{c}.
\label{eq:time_unit_growth_rate}
$
Using
\eqref{eq:decomposition_scale_invariant} and the change of variable $\widetilde a=ca$,
we obtain
\begin{equation*}
\int_0^\infty
e^{-\widetilde r\widetilde a}
\widetilde V(\widetilde a)\,
d\widetilde a
 =
\int_0^\infty
e^{-(r/c)\widetilde a}
\frac{1}{c}
V\!\left(\frac{\widetilde a}{c}\right)
d\widetilde a
\nonumber\\
 =
\int_0^\infty
e^{-ra}V(a)\,da
=
1.
\end{equation*}

\noindent The reproductive-age density therefore transforms as
\begin{equation}
\widetilde p(\widetilde a)
=
e^{-\widetilde r\widetilde a}
\widetilde V(\widetilde a)
=
\frac{1}{c}
p\!\left(\frac{\widetilde a}{c}\right)
\label{eq:time_unit_density}
\end{equation}
or, in differential form, 
$
\widetilde p(\widetilde a)\,d\widetilde a
=
p(a)\,da.
$
Hence \(p\) and \(\widetilde p\) represent the same reproductive-age
probability measure in different temporal coordinates.

Generation time scales as age:
\begin{equation}
\widetilde T
=
\int_0^\infty
\widetilde a\,
\widetilde p(\widetilde a)\,
d\widetilde a
=
c\int_0^\infty
a\,p(a)\,da
=
cT.
\label{eq:time_unit_generation_time}
\end{equation}

The reference duration defined in Section~\ref{sec:continuous_demography}
is unchanged by a change of units. Under \(\widetilde a=ca\), its numerical
value becomes
\begin{equation}
\widetilde\tau=c\tau.
\label{eq:time_unit_reference_scale}
\end{equation}
The two coordinate expressions for cycle entropy are thus

\begin{equation*}
S
=
-\int_0^\infty
p(a)\log\!\bigl(\tau p(a)\bigr)\,da
 \qquad {\rm and} \qquad
\widetilde S
=
-\int_0^\infty
\widetilde p(\widetilde a)
\log\!\bigl(
\widetilde\tau\widetilde p(\widetilde a)
\bigr)\,d\widetilde a.
\end{equation*}

\begin{theorem}[Transformation under a linear time scaling]
\label{pro:reference_scale_change}

Let \(\widetilde a=ca\), \(c>0\). Then the cycle entropy $S$, generation time $T$ and evolutionary entropy $H$ scale according to:
\begin{enumerate}
	\item[(i)] $\widetilde S=S,$
		\item[(ii)] $\widetilde T=cT,$
	\item[(iii)] $\widetilde H=\dfrac{H}{c}.$
	\end{enumerate}
\end{theorem}

\begin{proof}
From
\eqref{eq:time_unit_density} and
\eqref{eq:time_unit_reference_scale} it follows that
\[
\widetilde\tau\,
\widetilde p(\widetilde a)
=
c\tau
\frac{1}{c}
p\!\left(\frac{\widetilde a}{c}\right)
=
\tau
p\!\left(\frac{\widetilde a}{c}\right).
\]
Hence, after the change of variables \(\widetilde a=ca\), we obtain

\[
\begin{aligned}
	\widetilde S
	&=
	-\int_0^\infty
	\frac1c
	p\!\left(\frac{\widetilde a}{c}\right)
	\log\!\left(
	\tau
	p\!\left(\frac{\widetilde a}{c}\right)
	\right)
	\,d\widetilde a
	\\
	&=
	-\int_0^\infty
	p(a)\log(\tau p(a))\,da
	=
	S,
\end{aligned}
\]

\noindent establishing $(i)$.  Equation 
\eqref{eq:time_unit_generation_time} gives $(ii)$.  Finally, 
\begin{equation}
	\label{eq_EE_scaling}
\widetilde H
=
\frac{\widetilde S}{\widetilde T}
=
\frac{S}{cT}
=
\frac{H}{c},
\end{equation}
which proves  $(iii)$.
\end{proof}

The identity \eqref{eq_EE_scaling} expresses a change of time scale. The evolutionary entropy is unchanged once both values are
expressed in the same physical unit of time.
As an example, suppose that the generation time is \(4\) when age is
measured in five-year units. Expressed in years,
$
\widetilde T=5\times4=20.
$
If the cycle entropy is \(S\), the corresponding evolutionary entropy is
$
H=\frac{S}{4}
$
per five-year unit and
$
\widetilde H=\frac{S}{20}
$
per year. Since
$
\widetilde H=\frac{1}{5}H,
$
both expressions represent the same evolutionary entropy, differing
only in the unit of time. 

This coordinate change differs from the biological dilation considered in
\cite{DemetriusLegendreHarremoes2009}, where two life histories are compared in
a fixed age coordinate and the differential cycle entropy acquires a
logarithmic scaling term.

\begingroup
\begin{corollary}[Potential under changes of unit and reference]
\label{cor:potential_unit_reference}
\label{pro:reproductive_potential_invariance}
For the same demographic process expressed in the coordinates
\(\widetilde a=ca\), with \(c>0\) and the same reference duration
\(\widetilde\tau=c\tau\),
\begin{equation*}
\widetilde\Phi_{\widetilde\tau}=\frac{\Phi_\tau}{c}.
\end{equation*}
If instead the age coordinate is fixed and the reference duration is changed
from \(\tau_1>0\) to \(\tau_2>0\), then
\begin{equation*}
\begin{aligned}
S_{\tau_2}&=S_{\tau_1}-\log(\tau_2/\tau_1),\\
H_{\tau_2}&=H_{\tau_1}-\frac{\log(\tau_2/\tau_1)}{T},
\qquad
\Phi_{\tau_2}=\Phi_{\tau_1}+\frac{\log(\tau_2/\tau_1)}{T}.
\end{aligned}
\end{equation*}
The growth rate \(r\) is unchanged by this replacement of the reference.
\end{corollary}
The proof is given in Appendix~\ref{app:potential_proofs}.

\par\endgroup

\begin{rem}
	\label{rem:entropy_potential_covariance}

	A change of temporal unit does not alter the demographic description; it merely changes the coordinate used to measure age, for example from
	years to months. By contrast, temporal aggregation replaces a
	continuous reproductive-age distribution by probabilities associated
	with age classes of different  widths. This process may discard information about
	reproductive timing within each class. The aggregation problem is the
	subject of the next section.
\end{rem}

\section{Temporal aggregation and continuous reconstruction}
\label{sec:temporal_discretization}

We now address the continuous--discrete correspondences in
Table~\ref{tab:intro_questions}. We first study temporal aggregation,
identifying the information preserved and the information lost when a
continuous reproductive-age distribution is replaced by age classes.
We then consider the inverse problem: reconstructing the continuous
distribution from the discrete data.
Unless stated otherwise, age is measured in years.

\subsection{From continuous age to temporal classes}
\label{subsec:continuous_to_discrete}

Let \(h>0\) denote the class width, and partition
\([0,\infty)\) into the intervals
\begin{equation*}
	I_j^{(h)}
	=
	[jh,(j+1)h),
	\qquad
	j=0,1,2,\ldots.
\end{equation*}

\begin{rem}
	\label{rem:translated_temporal_grids}
	All the constructions below remain valid for a translated partition
	\[
	I_j^{(h,a_0)}
	=
	[a_0+jh,a_0+(j+1)h),
	\qquad
	j=0,1,2,\ldots,
	\]
	provided all meshes use the same origin \(a_0\) and \(p=0\) almost everywhere below \(a_0\). For simplicity we take here 
	\(a_0=0\). In the constant-hazard model of
	Section~\ref{sec:constant_hazard}, however, reproduction begins at age
	\(A\), so the natural choice will be \(a_0=A\).
\end{rem}

The reproductive probability assigned to the \(j\)-th class is
\begin{equation*}
	P_j^{(h)}
	=
	\int_{I_j^{(h)}}p(a)\,da.
\end{equation*}
Since \(p\) is a probability density, we have 
$	\sum_{j=0}^\infty P_j^{(h)}=1.$

The probabilities \(P_j^{(h)}\) determine the total reproductive
contribution in each class, but not its distribution within the class.
Accordingly, for every \(P_j^{(h)}>0\), we define the conditional density
\begin{equation}
	\varphi_j^{(h)}(a)
	=
	\frac{p(a)}{P_j^{(h)}},
	\qquad
	a\in I_j^{(h)}.
	\label{eq:general_conditional_class_density}
\end{equation}
Note that 
$
	\int_{I_j^{(h)}} \varphi_j^{(h)}(a) \,da = 1.
	$
Also, from the expression
\begin{equation}
	p(a)
	=
	P_j^{(h)}\varphi_j^{(h)}(a),
	\qquad
	a\in I_j^{(h)}
	\label{eq:general_class_factorization}
\end{equation}
it is clear that 
\(P_j^{(h)}\) describes the distribution of reproductive
contribution between classes, whereas \(\varphi_j^{(h)}\) describes its
distribution within each class. The continuous-to-discrete process of temporal aggregation 
\[
p(a)
\longmapsto
\left(P_j^{(h)}\right)_{j\ge0},
\]
therefore discards the family of conditional densities
$
\left(\varphi_j^{(h)}\right)_{j\ge0},
$ keeping only the probability mass of each class.
This also explains why a coarse age-class model cannot, in general, be
refined uniquely: \(P_j^{(h)}\) determines the total reproductive
probability in \(I_j^{(h)}\), but not its allocation among finer
subintervals.

\subsection{From temporal classes to continuous quantities}
\label{subsec:discrete_to_continuous}
\label{sec:general_exact_partition_reconstruction}

We now consider the inverse problem. The class probabilities
\(P_j^{(h)}\) do not, in general, determine the continuous
reproductive-age density. The missing information is precisely the
family of within-class distributions introduced above.

The Shannon entropy of the class probabilities is
\begin{equation}
S_h
=
-\sum_{j=0}^\infty
P_j^{(h)}\log P_j^{(h)},
\label{eq:discrete_Shannon_entropy}
\end{equation}
where the standing convention  \(0\log0=0\) is assumed.

For each non-empty class, define the within-class cycle entropy by
\begin{equation}
C_j^{(h)}
=
-
\int_{I_j^{(h)}}
\varphi_j^{(h)}(a)
\log\!\bigl(\tau \varphi_j^{(h)}(a)\bigr)\,da,
\label{eq:general_within_class_entropy}
\end{equation}
where \(\tau\) denotes the numerical value of the fixed one-year
reference duration. Since age is measured in years throughout this section, we have
\(\tau=1\).

The following chain-rule identity \cite{Shannon1948} gives the exact relation between continuous
cycle entropy and temporal aggregation.

\begin{theorem}[Entropy decomposition under temporal aggregation]
\label{thm:general_exact_entropy_reconstruction}

For every temporal partition for which the relevant entropies are
finite,
\begin{equation}
S
=
S_h
+
\sum_{j=0}^\infty
P_j^{(h)}C_j^{(h)}.
\label{eq:general_exact_entropy_partition}
\end{equation}
\end{theorem}

\begin{proof}
For \(a\in I_j^{(h)}\),
\eqref{eq:general_class_factorization} gives
\[
\log\!\bigl(\tau p(a)\bigr)
=
\log P_j^{(h)}
+
\log\!\bigl(\tau \varphi_j^{(h)}(a)\bigr).
\]
Hence
\begin{align*}
S
&=
-
\sum_{j=0}^{\infty}
\int_{I_j^{(h)}}
p(a)\log\!\bigl(\tau p(a)\bigr)\,da
\\
&=
-
\sum_{j=0}^{\infty}
P_j^{(h)}
\int_{I_j^{(h)}}
\varphi_j^{(h)}(a)
\left[
\log P_j^{(h)}
+
\log\!\bigl(\tau \varphi_j^{(h)}(a)\bigr)
\right]
\,da
\\
&=
-
\sum_{j=0}^{\infty}
P_j^{(h)}\log P_j^{(h)}
+
\sum_{j=0}^{\infty}
P_j^{(h)}C_j^{(h)},
\end{align*}
which, by \eqref{eq:discrete_Shannon_entropy},  is
\eqref{eq:general_exact_entropy_partition}.
\end{proof}

The decomposition
\eqref{eq:general_exact_entropy_partition}
separates the contributions from between and within the temporal
classes. The term \(S_h\) is the Shannon entropy of the class
probabilities, whereas the term 
$
\sum_{j=0}^\infty
P_j^{(h)}C_j^{(h)}
$
accounts for the within-class contribution.
Consequently, the class probabilities alone do not determine the
continuous cycle entropy. Exact reconstruction requires the
within-class entropies \(C_j^{(h)}\), which can be computed from the conditional
densities \(\varphi_j^{(h)}\); the entropies alone do not determine those densities.

Generation time admits an analogous decomposition. For each
\(P_j^{(h)}>0\), define the conditional mean reproductive age by
\begin{equation}
	\bar a_j^{(h)}
	=
	\frac{
		\displaystyle
		\int_{I_j^{(h)}}a\,p(a)\,da
	}{
		P_j^{(h)}
	}
	=
	\int_{I_j^{(h)}}
	a\,\varphi_j^{(h)}(a)\,da.
	\label{eq:class_mean_age}
\end{equation}

\begin{proposition}[Exact reconstruction of generation time]
	\label{pro:generation_time_exact_invariance}
	\label{pro:general_exact_generation_reconstruction}
	
	For every \(h>0\),
	\begin{equation}
			T
			=
			\sum_{j=0}^\infty
			P_j^{(h)}\bar a_j^{(h)}.
		\label{eq:generation_time_exact_invariance}
	\end{equation}
\end{proposition}

\begin{proof}
Using \eqref{eq:class_mean_age},
\begin{align*}
\sum_{j=0}^\infty
P_j^{(h)}\bar a_j^{(h)}
&=
\sum_{j=0}^\infty
\int_{I_j^{(h)}}a\,p(a)\,da
\\
&=
\int_0^\infty a\,p(a)\,da
=
T.
\end{align*}
\end{proof}

For comparison with a conventional  age-class model, let $
a_j^{\mathrm{rep}}\in I_j^{(h)}
$
be a representative age for the \(j\)-th class and define
the mean within-class displacement
\begin{equation*}
\delta_j^{(h)}
=
\int_{I_j^{(h)}}
\left(
a-a_j^{\mathrm{rep}}
\right)
\varphi_j^{(h)}(a)\,da.
\end{equation*}
Then
$
\bar a_j^{(h)}
=
a_j^{\mathrm{rep}}
+
\delta_j^{(h)},$
so that
\begin{equation}
T
=
\sum_{j=0}^\infty
P_j^{(h)}
\left(
a_j^{\mathrm{rep}}
+
\delta_j^{(h)}
\right).
\label{eq:general_exact_generation_partition}
\end{equation}

Combining
\eqref{eq:general_exact_entropy_partition} and
\eqref{eq:general_exact_generation_partition} gives the exact
finite-partition representation
\begin{equation}
H
=
\frac{
S_h
+
\displaystyle\sum_{j=0}^\infty
P_j^{(h)}C_j^{(h)}
}{
\displaystyle\sum_{j=0}^\infty
P_j^{(h)}
\left(
a_j^{\mathrm{rep}}
+
\delta_j^{(h)}
\right)
}.
\label{eq:general_exact_H_partition}
\end{equation}

All ages in
\eqref{eq:general_exact_H_partition} are  measured in
years. Hence the denominator is the physical generation time,
and the resulting evolutionary entropy is expressed per year,
independently of the class width \(h\).

The preceding formulas are expressed in physical units. In a
conventional age-class model, however, representative ages are often
given in class units.  Write
\begin{equation*}
a_j^{\mathrm{rep}}
=
h n_j,
\end{equation*}
where \(n_j\) is the representative age measured in class units, and
define
\begin{equation*}
\varepsilon_j^{(h)}
=
\frac{\delta_j^{(h)}}{h}.
\end{equation*}
Then
$
a_j^{\mathrm{rep}}
+
\delta_j^{(h)}
=
h
\left(
n_j+\varepsilon_j^{(h)}
\right).
$ Hence the generation time in class units is 
\begin{equation}
T_h^{\mathrm{class}}
=
\sum_{j=0}^\infty
P_j^{(h)}
\left(
n_j+\varepsilon_j^{(h)}
\right),
\label{eq:generation_time_class_units}
\end{equation}

so that $T=hT_h^{\mathrm{class}}.$
Substituting into
\eqref{eq:general_exact_H_partition}
gives
\begin{equation}
			H
		=
		\frac1h
		\frac{
			S_h
			+
			\displaystyle\sum_{j=0}^\infty
			P_j^{(h)}C_j^{(h)}
		}{
			T_h^{\mathrm{class}}
		}.
	\label{eq:evolutionary_entropy_class_units}
\end{equation}

\begingroup
\begin{corollary}[Exact reconstruction of reproductive potential]
\label{cor:potential_exact_reconstruction}
Consider a temporal partition of the same reproductive process \(V\),
retaining its Euler--Lotka rate \(r\). Let \(P_j^{(h)}\), \(C_j^{(h)}\)
and \(T_h^{\mathrm{class}}\) be the class probabilities, within-class cycle
entropies and generation time in class units defined above, with the fixed
one-year reference. The reconstructed potential is
\begin{equation}
\Phi_h^{\mathrm{rec}}
:=r-\frac{S_h+\sum_jP_j^{(h)}C_j^{(h)}}{hT_h^{\mathrm{class}}}
=\Phi.
\label{eq:potential_exact_reconstruction}
\end{equation}
Thus the same within-class entropies and conditional mean ages that
recover \(H\) also recover \(\Phi\), provided \(r\) is retained.
\end{corollary}
The proof is given in Appendix~\ref{app:potential_proofs}.

\par\endgroup

Equation \eqref{eq:evolutionary_entropy_class_units} shows how the
temporal unit enters a discrete representation. The quotient
\[
\frac{
S_h+\sum_jP_j^{(h)}C_j^{(h)}
}{
T_h^{\mathrm{class}}
}
\]
is the evolutionary entropy rate per class interval. Division by \(h\)
expresses the same rate per year.

To provide a specific example, a generation time of four classes with \(h=5\) years is
$
T=5\times4=20\ {\rm years}.
$
The equivalent annual model has \(h=1\) and a generation time of
twenty classes. Since both models reconstruct the same cycle entropy,
equation \eqref{eq:evolutionary_entropy_class_units} yields the same
evolutionary entropy per year.

In summary, the continuous-to-discrete passage
\[
p
\longmapsto
\left(P_j^{(h)}\right)_{j\ge0}
\]
retains only the reproductive probability assigned to each temporal
class. The within-class distributions \(\varphi_j^{(h)}\) are discarded. In the reverse direction, 
the pair
\[
\left\{
P_j^{(h)},\,\varphi_j^{(h)}
\right\}_{j\ge0}
\]
determines the original density \(p\) through
\eqref{eq:general_class_factorization}. Thus temporal aggregation is
exactly reversible only if the within-class distributions are retained.
Reconstructing the full density is, however, stronger than necessary if the goal is 
to recover only the quantities \(S\), \(T\), and \(H\): in that case it suffices to
know the within-class entropies \(C_j^{(h)}\) and the conditional mean
ages \(\bar a_j^{(h)}\), or equivalently the displacements
\(\delta_j^{(h)}\).

Section~\ref{sec:unresolved_within_class_structure} considers the common situation in which the
within-class structure is unavailable. 

\section{Change of temporal mesh}
\label{sec:change_temporal_mesh}
\label{sec:exact_commutative_principle}

We now consider the discrete-to-discrete problem in
Table~\ref{tab:intro_questions}. Suppose the same reproductive process
is represented on two different, uniformly nested temporal meshes.
The demographic process is unchanged; only the
temporal partition differs.

The ordinary Shannon entropy of the class probabilities is generally
different on the two meshes. The relevant question is whether the two
representations nevertheless reconstruct the same cycle entropy, the
same physical generation time, and hence the same evolutionary entropy
per year.

We denote by $h_f$ the fine mesh and by $h_c$ the coarse mesh, so that 
$
	h_c=mh_f
$
for some $m \in \mathbb{N}$. Every class of width \(h_c\) is the union of \(m\) consecutive
classes of width \(h_f\):
\begin{equation*}
	I_j^{(h_c)}
	=
	\bigcup_{s=0}^{m-1}
	I_{jm+s}^{(h_f)}.
\end{equation*}

\subsection{Coarsening and refinement}
\label{subsec:coarsening_refinement}

Passing from the finer mesh \(h_f\) to the coarser mesh \(h_c\) is
immediate. The reproductive probabilities satisfy
\begin{equation*}
		P_j^{(h_c)}
		=
		\sum_{s=0}^{m-1}
		P_{jm+s}^{(h_f)}.
\end{equation*}
Thus coarsening is determined entirely by the finer class
probabilities.

The reverse operation is not determined by the coarse probabilities
alone. For each class with \(P_j^{(h_c)}>0\), define the splitting
coefficients
\begin{equation*}
	\alpha_{j,s}
	=
	\frac{
		P_{jm+s}^{(h_f)}
	}{
		P_j^{(h_c)}
	},
	\qquad
	s=0,\ldots,m-1.
\end{equation*}
Then
\begin{equation}
	\alpha_{j,s}\ge0,
	\qquad
	\sum_{s=0}^{m-1}\alpha_{j,s}=1,
	\label{eq:mesh_splitting_normalization}
\end{equation}
and
\begin{equation}
		P_{jm+s}^{(h_f)}
		=
		P_j^{(h_c)}\alpha_{j,s}.
	\label{eq:exact_mesh_refinement}
\end{equation}

The coefficients \(\alpha_{j,s}\) specify how the reproductive
probability of the \(j\)-th coarse class is distributed among its
\(m\) finer subclasses. The coarse probabilities
\(\bigl(P_j^{(h_c)}\bigr)_{j\ge0}\) do not determine these
coefficients. Hence a coarse discrete model does not generally
determine its own refinement.

If both discrete representations arise from the same continuous
reproductive-age density, then the splitting coefficients are
determined by the conditional densities of
Section~\ref{subsec:continuous_to_discrete}:
\begin{equation}
		\alpha_{j,s}
		=
		\int_{I_{jm+s}^{(h_f)}}
		\varphi_j^{(h_c)}(a)\,da.
	\label{eq:splitting_coefficients_from_conditional_density}
\end{equation}

Thus the within-class information identified in
Section~\ref{sec:temporal_discretization} is exactly the information
needed to reconstruct a refinement consistent with the same
demographic process.

\subsection{Cycle entropy under a change of mesh}
\label{subsec:entropy_mesh_change}

The splitting coefficients determine the relation between the 
Shannon entropies of the coarse and fine class distributions.  For each coarse
class, define the {\em refinement entropy} by
\begin{equation*}
	E_j^{(h_c,h_f)}
	=
	-\sum_{s=0}^{m-1}
	\alpha_{j,s}\log\alpha_{j,s}.
\end{equation*}

\begin{theorem}[Entropy decomposition under mesh refinement]
	\label{thm:entropy_mesh_refinement}
	
	Let \(h_c=mh_f\). Then
	\begin{equation}
			S_{h_f}
			=
			S_{h_c}
			+
			\sum_{j=0}^\infty
			P_j^{(h_c)}
			E_j^{(h_c,h_f)}.
		\label{eq:entropy_mesh_refinement}
	\end{equation}
\end{theorem}

\begin{proof}
	Using \eqref{eq:exact_mesh_refinement},
	\begin{align*}
		S_{h_f}
		&=
		-\sum_{j=0}^\infty
		\sum_{s=0}^{m-1}
		P_j^{(h_c)}
		\alpha_{j,s}
		\log\!\left(
		P_j^{(h_c)}\alpha_{j,s}
		\right)
\\
		&=
		-\sum_{j=0}^\infty
		P_j^{(h_c)}
		\log P_j^{(h_c)}
		\sum_{s=0}^{m-1}\alpha_{j,s}
\\
		&\quad
		-
		\sum_{j=0}^\infty
		P_j^{(h_c)}
		\sum_{s=0}^{m-1}
		\alpha_{j,s}\log\alpha_{j,s}.
\end{align*}
	Using \eqref{eq:mesh_splitting_normalization} and
	\eqref{eq:discrete_Shannon_entropy} yields
	\eqref{eq:entropy_mesh_refinement}, finishing the proof.
\end{proof}

Theorem~\ref{thm:entropy_mesh_refinement} is the discrete analogue of
	Theorem~\ref{thm:general_exact_entropy_reconstruction}.
	Informally, it states that the fine Shannon entropy equals the sum of the coarse Shannon entropy with the refinement entropy.

 Refining the mesh resolves the conditional entropy within each coarse
class. Since $0
\leq
E_j^{(h_c,h_f)}
\leq
\log m,$ it follows that 

\begin{equation}
		S_{h_c}
		\leq
		S_{h_f}
		\leq
		S_{h_c}+\log m.
	\label{eq:entropy_refinement_bounds}
\end{equation}

Thus Shannon entropy is not mesh invariant. Except in
degenerate cases, refining the mesh $h$ changes \(S_h\), even though the
two discrete models represent the same demographic process.

The missing contribution is the within-class cycle entropy.
From \eqref{eq:general_conditional_class_density} and  \eqref{eq:splitting_coefficients_from_conditional_density}, we may write 
\[ 	\int_{I_{jm+s}^{(h_f)}}
\varphi_j^{(h_c)}(a)\,da = \frac{1}{P_j^{(h_c)}} \int_{I_{jm+s}^{(h_f)}}
p(a)\,da 
=
\frac{
	P_{jm+s}^{(h_f)}
}{
	P_j^{(h_c)}
} = 	\alpha_{j,s}.
\]
and therefore, for $a \in I_{jm+s}^{(h_f)} \subset I_{j}^{(h_c)}$, 
\[
\varphi_j^{(h_c)}(a) = \frac{
	P_{jm+s}^{(h_f)}
}{
	P_j^{(h_c)} 
} \varphi_{jm+s}^{(h_f)}(a)= 	\alpha_{j,s} \, \varphi_{jm+s}^{(h_f)}(a).
\]
 Recalling from \eqref{eq:general_within_class_entropy} 
the definition of within-class cycle entropy $C_j^{(h)}$, it follows that 

\begin{equation}
C_j^{(h_c)}
=
E_j^{(h_c,h_f)}
+
\sum_{s=0}^{m-1}
\alpha_{j,s}
C_{jm+s}^{(h_f)}.
\label{eq:within_class_entropy_mesh_decomposition}
\end{equation}

Thus the entropy of a coarse class decomposes into the entropy of its
subdivision and the mean entropy within the finer subclasses. 

Define the cycle entropy reconstructed from a mesh of width \(h\) by
\begin{equation*}
S_h^{\mathrm{rec}}
=
S_h
+
\sum_{j=0}^\infty
P_j^{(h)}C_j^{(h)}.
\end{equation*}

\begin{theorem}[Mesh invariance of reconstructed cycle entropy]
\label{thm:reconstructed_entropy_mesh_invariance}

For two nested representations of the same
reproductive process,
\begin{equation*}
S_{h_c}^{\mathrm{rec}}
=
S_{h_f}^{\mathrm{rec}}
=
S.
\end{equation*}
\end{theorem}

\begin{proof}
Multiplying
\eqref{eq:within_class_entropy_mesh_decomposition} by
\(P_j^{(h_c)}\), summing over \(j\), and applying 
\eqref{eq:exact_mesh_refinement}, we obtain

\[
\sum_{j=0}^\infty
P_j^{(h_c)}C_j^{(h_c)}
=
\sum_{j=0}^\infty
P_j^{(h_c)}E_j^{(h_c,h_f)}
+
\sum_{j=0}^\infty
\sum_{s=0}^{m-1}
P_{jm+s}^{(h_f)}
C_{jm+s}^{(h_f)}.
\]
Equation \eqref{eq:entropy_mesh_refinement} then gives
\[
S_{h_c}
+
\sum_jP_j^{(h_c)}C_j^{(h_c)}
=
S_{h_f}
+
\sum_kP_k^{(h_f)}C_k^{(h_f)}.
\]
By Theorem~\ref{thm:general_exact_entropy_reconstruction}, both sides
 equal the continuous cycle entropy \(S\).
\end{proof}

The essential distinction is that, although $S_{h_c}\neq S_{h_f}$ in general, 
$
S_{h_c}^{\mathrm{rec}}
=
S_{h_f}^{\mathrm{rec}}
$ always holds.
Thus the change in Shannon entropy under mesh refinement is exactly
compensated by the corresponding change in the within-class entropy.

\subsection{Generation time under a change of mesh}
\label{subsec:generation_mesh_change}

The same argument applies to generation time. The exact generation
time obtained from the finer mesh is
\begin{equation*}
T_{h_f}
=
\sum_{j=0}^\infty
\sum_{s=0}^{m-1}
P_{jm+s}^{(h_f)}
\bar a_{jm+s}^{(h_f)}.
\end{equation*}
For every coarse class with \(P_j^{(h_c)}>0\), the conditional mean
age satisfies
\begin{equation}
\bar a_j^{(h_c)}
=
\sum_{s=0}^{m-1}
\alpha_{j,s}
\bar a_{jm+s}^{(h_f)}.
\label{eq:aggregation_consistent_age}
\end{equation}

\begin{proposition}[Generation time is preserved under exact coarsening]
\label{pro:generation_time_mesh_coarsening}
For two equivalent nested temporal representations,
\begin{equation*}
T_{h_c}=T_{h_f}=T.
\end{equation*}
\end{proposition}

\begin{proof}
Using
\eqref{eq:aggregation_consistent_age},
\eqref{eq:exact_mesh_refinement}, and
\eqref{eq:generation_time_exact_invariance},
\begin{align*}
T_{h_c}
&=
\sum_{j=0}^\infty
P_j^{(h_c)}
\bar a_j^{(h_c)}
\\
&=
\sum_{j=0}^\infty
P_j^{(h_c)}
\sum_{s=0}^{m-1}
\alpha_{j,s}
\bar a_{jm+s}^{(h_f)}
\\
&=
\sum_{j=0}^\infty
\sum_{s=0}^{m-1}
P_{jm+s}^{(h_f)}
\bar a_{jm+s}^{(h_f)}
\\
&=
T_{h_f}
=
T.
\end{align*}
\end{proof}

The quantities \(T_{h_c}\) and \(T_{h_f}\) are expressed in years. Measured in class units,
Section~\ref{sec:temporal_discretization} gives
\begin{equation}
h_cT_{h_c}^{\mathrm{class}}
=
h_fT_{h_f}^{\mathrm{class}}
=
T.
\label{eq:mesh_generation_time_class_units}
\end{equation}

\subsection{Invariance of evolutionary entropy}
\label{subsec:evolutionary_entropy_mesh_invariance}

We now state the discrete-to-discrete invariance principle.

\begin{theorem}[Mesh invariance of evolutionary entropy]
\label{thm:evolutionary_entropy_mesh_invariance}

Let two nested representations describe the
same reproductive process and retain the within-class information
required for exact reconstruction. Then
\begin{equation*}
\frac{1}{h_c}
\frac{
S_{h_c}^{\mathrm{rec}}
}{
T_{h_c}^{\mathrm{class}}
}
=
\frac{1}{h_f}
\frac{
S_{h_f}^{\mathrm{rec}}
}{
T_{h_f}^{\mathrm{class}}
}
=
H.
\end{equation*}
\end{theorem}

\begin{proof}
By
Theorem~\ref{thm:reconstructed_entropy_mesh_invariance},
$
S_{h_c}^{\mathrm{rec}}
=
S_{h_f}^{\mathrm{rec}}
=
S.
$
 Moreover, from  \eqref{eq:mesh_generation_time_class_units},
$
h_cT_{h_c}^{\mathrm{class}}
=
h_fT_{h_f}^{\mathrm{class}}
=
T.
$
Therefore
\[
\frac{1}{h_i}
\frac{
S_{h_i}^{\mathrm{rec}}
}{
T_{h_i}^{\mathrm{class}}
}
=
\frac{S}{T}
=
H,
\qquad i\in\{c,f\}.
\]

\end{proof}

The discrete quantities
$
S_h
\ \text{and} \
\frac{S_h}{T_h^{\mathrm{class}}}
$
are not mesh invariant: the first changes because refinement resolves
additional within-class structure, while the second is measured per
class interval. By Theorem~\ref{thm:evolutionary_entropy_mesh_invariance},
the quantity
$
\frac{1}{h}
\frac{S_h^{\mathrm{rec}}}{T_h^{\mathrm{class}}}
$
is mesh invariant. Once cycle entropy is reconstructed and the entropy
rate is expressed in a common unit of time, equivalent representations
have the same evolutionary entropy.

For example, the same human population represented by five-year or
annual age classes must have the same evolutionary entropy per year,
provided the two models represent the same reproductive process. Any
difference therefore reflects either discarded within-class
information or a genuinely different demographic model, not a
dependence of evolutionary entropy on class width.

\begingroup
\begin{corollary}[Mesh invariance of reproductive potential]
\label{cor:potential_mesh_invariance}
Let two nested temporal representations describe the same process \(V\),
with the same Euler--Lotka rate \(r\), time unit and reference duration.
If both retain the within-class entropies and conditional mean ages, then
\begin{equation*}
\Phi_{h_c}^{\mathrm{rec}}=\Phi_{h_f}^{\mathrm{rec}}=\Phi.
\end{equation*}
\end{corollary}
The proof is given in Appendix~\ref{app:potential_proofs}.

\par\endgroup

\subsection{Uniform refinement and the logarithmic correction}
\label{subsec:uniform_refinement_correction}

Exact reconstruction requires within-class temporal information. When
it is unavailable, the simplest approximation is to assume that
reproductive contribution is uniformly distributed within each class. Refining each class into \(m\) equal subclasses then gives
\begin{equation*}
\alpha_{j,s}
=
\frac{1}{m}.
\end{equation*}
Hence \(E_j^{(h_c,h_f)}=\log m\) for every \(j\), so
\begin{equation*}
S_{h_f}
=
S_{h_c}+\log m.
\end{equation*}
Since \(h_c=mh_f\), we have 
$
\log m
=
\log\frac{h_c}{h_f},
$
and therefore
\begin{equation*}
S_{h_f}+\log h_f
=
S_{h_c}+\log h_c.
\end{equation*}

Thus, when the within-class structure is approximated as uniform, the
logarithmic correction compensates for the entropy increase produced by
mesh refinement, and the quantity \(S_h+\log h\) is invariant.

\section{Unresolved within-class structure and the logarithmic correction}
\label{sec:unresolved_within_class_structure}

The exact reconstruction results of
Sections~\ref{sec:temporal_discretization} and
\ref{sec:change_temporal_mesh} require information on reproductive
timing within each class. In most discrete demographic models,
however, only the class probabilities are available; the within-class
distributions \(\varphi_j^{(h)}\) are not.

The simplest continuous reconstruction distributes the reproductive
probability of each class uniformly over its interval. Thus,
for \(a\in I_j^{(h)}\), define
\begin{equation*}
p_h^{\mathrm{unif}}(a)
=
\frac{P_j^{(h)}}{h}.
\end{equation*}
This is the piecewise constant probability density associated with the
discrete class probabilities. Since
\[
\int_{I_j^{(h)}}
p_h^{\mathrm{unif}}(a)\,da
=
P_j^{(h)},
\]
this reconstruction preserves the reproductive probability
of every class.

The cycle entropy of the reconstructed density is
\begin{align*}
S_h^{\mathrm{unif}}
&=
-\int_0^\infty
p_h^{\mathrm{unif}}(a)
\log p_h^{\mathrm{unif}}(a)\,da
\\
&=
-\sum_{j=0}^\infty
\int_{I_j^{(h)}}
\frac{P_j^{(h)}}{h}
\log\!\left(
\frac{P_j^{(h)}}{h}
\right)da
\\
&=
-\sum_{j=0}^\infty
P_j^{(h)}
\log P_j^{(h)}
+
\log h.
\end{align*}
Hence
\begin{equation}
S_h^{\mathrm{unif}}
=
S_h+\log h.
\label{eq:uniform_reconstructed_entropy}
\end{equation}
Thus the logarithmic term is not an ad hoc correction to the discrete
Shannon entropy: it is the within-class entropy of the piecewise
uniform reconstruction.
Indeed, 
$
\varphi_j^{(h)}(a)=\frac{1}{h}$
for $
a\in I_j^{(h)},
$
and therefore
\[
C_j^{(h)}
=
-\int_{I_j^{(h)}}
\frac{1}{h}
\log\frac{1}{h}\,da
=
\log h.
\]
Equation~\eqref{eq:general_exact_entropy_partition} therefore gives
$
S=S_h+\log h
$
whenever \(p\) is uniform on every class.

For a general density, the deviation from uniformity is measured by
\begin{equation*}
D_j^{(h)}
=
\int_{I_j^{(h)}}
\varphi_j^{(h)}(a)
\log\!\left(
h \varphi_j^{(h)}(a)
\right)da.
\end{equation*}
This is the relative entropy of the true conditional density
\(\varphi_j^{(h)}\) with respect to the uniform density \(1/h\) on
\(I_j^{(h)}\) \cite{KullbackLeibler1951}. Hence
$
D_j^{(h)}\geq0,
$
with equality if and only if \(\varphi_j^{(h)}=1/h\) almost everywhere.
Consequently,
\begin{equation*}
C_j^{(h)}
=
\log h-D_j^{(h)}.
\end{equation*}
Substituting this identity into
\eqref{eq:general_exact_entropy_partition} gives the exact formula
\begin{equation}
S_h+\log h
=
S
+
\sum_{j=0}^\infty
P_j^{(h)}D_j^{(h)}.
\label{eq:corrected_entropy_exact_defect}
\end{equation}

The error after the logarithmic correction in  \eqref{eq:corrected_entropy_exact_defect} admits a
global interpretation. Since, for \(a\in I_j^{(h)}\),
\[
p_h^{\mathrm{unif}}(a)
=
\frac{P_j^{(h)}}{h},
\qquad
p(a)
=
P_j^{(h)}\varphi_j^{(h)}(a),
\]
we have
\[
\log
\frac{p(a)}{p_h^{\mathrm{unif}}(a)}
=
\log\!\left(
h \varphi_j^{(h)}(a)
\right).
\]
Therefore
\begin{equation}
S_h+\log h-S
=
D_{\mathrm{KL}}
\!\left(
p\,\middle\|\,p_h^{\mathrm{unif}}
\right),
\label{eq:corrected_entropy_global_KL}
\end{equation} where
$
D_{\mathrm{KL}}
\!\left(
p\,\middle\|\,p_h^{\mathrm{unif}}
\right)
=
\int_0^\infty
p(a)
\log
\frac{p(a)}{p_h^{\mathrm{unif}}(a)}
\,da
$  is the Kullback--Leibler divergence (relative entropy). 
Equivalently,
\[
D_{\mathrm{KL}}
\!\left(
p\,\middle\|\,p_h^{\mathrm{unif}}
\right)
=
\sum_{j=0}^\infty
P_j^{(h)}D_j^{(h)}.
\]
Thus the excess \(S_h+\log h-S\), not the term \(\log h\), is exactly the Kullback--Leibler divergence between the true reproductive-age density and its piecewise uniform reconstruction. This non-negative error cannot increase under nested refinement, since by \eqref{eq:entropy_refinement_bounds} the quantity \(S_h+\log h\) cannot increase when the mesh is refined.
In particular,
\begin{equation*}
S_h+\log h\geq S,
\end{equation*}
with equality iff \(p\) is constant almost everywhere on each non-empty
class.

\subsection{Continuous limit}
\label{subsec:corrected_entropy_continuous_limit}

Equation \eqref{eq:corrected_entropy_global_KL} identifies the corrected entropy error with the divergence between the true density and its piecewise uniform reconstruction. We now show that this error vanishes as the mesh is refined under the following assumptions.

\begin{theorem}[Continuous limit of the corrected discrete entropy]
\label{thm:corrected_entropy_continuous_limit}

Assume that \(p:[0,\infty)\to[0,\infty)\) is continuous and has compact
support. Then
\begin{equation}
\lim_{h\to0}
\left(
S_h+\log h
\right)
=
S.
\label{eq:corrected_entropy_continuous_limit}
\end{equation}
Equivalently,
\begin{equation}
\lim_{h\to0}
D_{\mathrm{KL}}
\!\left(
p\,\middle\|\,p_h^{\mathrm{unif}}
\right)
=
0.
\label{eq:within_class_relative_entropy_limit}
\end{equation}
\end{theorem}

\begin{proof}
Since \(p\) is continuous and compactly supported on \([0,\infty)\), it is uniformly continuous. Let
$
\omega_p(h)
=
\sup_{\lvert a-b\rvert\leq h}
\lvert p(a)-p(b)\rvert
$
be its modulus of continuity. Then
$
\omega_p(h)\longrightarrow0
\ 
\text{as } h\to0.
$

For \(a\in I_j^{(h)}\),
\begin{align*}
\left|
p_h^{\mathrm{unif}}(a)-p(a)
\right|
&=
\left|
\frac{1}{h}
\int_{I_j^{(h)}}p(s)\,ds
-
p(a)
\right|
\\
&\leq
\frac{1}{h}
\int_{I_j^{(h)}}
\left|
p(s)-p(a)
\right|ds
\\
&\leq
\omega_p(h).
\end{align*}
Therefore
\begin{equation}
\left\|
p_h^{\mathrm{unif}}-p
\right\|_\infty
\longrightarrow0.
\label{eq:uniform_histogram_convergence}
\end{equation}

Let
$\psi(x)=x\log x$ for $x \geq 0$,  with the usual understanding $0 \log 0 = 0$. 
Since \(p\) is bounded, so are the functions
\(p_h^{\mathrm{unif}}\).
Hence \(\psi\) is uniformly continuous on a common bounded interval
containing all their values.
 It then follows
from \eqref{eq:uniform_histogram_convergence} that
$
\psi\!\left(p_h^{\mathrm{unif}}\right)
\longrightarrow
\psi(p)
$
uniformly.

If \(\operatorname{supp}(p)\subset[0,L]\), then, for
\(0<h\le1\),
$
\operatorname{supp}(p_h^{\mathrm{unif}})
\subset
[0,L+1].
$
Therefore
\[
\int_0^\infty
p_h^{\mathrm{unif}}(a)
\log p_h^{\mathrm{unif}}(a)\,da
\longrightarrow
\int_0^\infty
p(a)\log p(a)\,da.
\]
Using \eqref{eq:uniform_reconstructed_entropy},
\[
S_h+\log h
=
-\int_0^\infty
p_h^{\mathrm{unif}}(a)
\log p_h^{\mathrm{unif}}(a)\,da,
\]
which proves
\eqref{eq:corrected_entropy_continuous_limit}.
Finally,
\eqref{eq:corrected_entropy_global_KL} gives
\[
D_{\mathrm{KL}}
\!\left(
p\,\middle\|\,p_h^{\mathrm{unif}}
\right)
=
S_h+\log h-S
\longrightarrow0,
\]
proving
\eqref{eq:within_class_relative_entropy_limit}.
\end{proof}

\begin{rem}
The compact-support assumption is convenient rather than essential. The same conclusion remains valid for densities on unbounded age
intervals, provided the histogram approximations converge with adequate
control of the entropy tails. The constant-hazard model of the next section has unbounded support, but
is treated exactly and does not require
Theorem~\ref{thm:corrected_entropy_continuous_limit}.
\end{rem}

\subsection{Generation time and evolutionary entropy}
\label{subsec:corrected_evolutionary_entropy_limit}

Under the uniform reconstruction, the representative reproductive age
of each class is its midpoint:
\begin{equation*}
a_j^{\mathrm{mid}}
=
\left(j+\frac12\right)h.
\end{equation*}
The corresponding generation time is
\begin{equation*}
T_h^{\mathrm{unif}}
=
\sum_{j=0}^\infty
P_j^{(h)}
\left(j+\frac12\right)h.
\end{equation*}

Since both the true conditional mean and the midpoint lie in the same
interval of width \(h\), we have 
$
\left|
T_h^{\mathrm{unif}}-T
\right|
\leq
\frac{h}{2}.
$
Consequently,
\begin{equation*}
\lim_{h\to0}
T_h^{\mathrm{unif}}
=
T.
\end{equation*}

Define the evolutionary entropy of the uniform reconstruction by
\begin{equation*}
H_h^{\mathrm{unif}}
=
\frac{
S_h+\log h
}{
T_h^{\mathrm{unif}}
}.
\end{equation*}

\begin{corollary}[Continuous limit of evolutionary entropy]
\label{cor:corrected_evolutionary_entropy_limit}

Under the assumptions of
Theorem~\ref{thm:corrected_entropy_continuous_limit},
\begin{equation*}
\lim_{h\to0}
H_h^{\mathrm{unif}}
=
H.
\end{equation*}
\end{corollary}

\begingroup
\begin{corollary}[Continuous limit of reconstructed potential]
\label{cor:potential_uniform_limit}
Under the assumptions of
Theorem~\ref{thm:corrected_entropy_continuous_limit}, retain the
Euler--Lotka rate \(r\) of the original process and the one-year reference.
For the uniform reconstruction, define
\begin{equation*}
\Phi_h^{\mathrm{unif}}:=r-H_h^{\mathrm{unif}}
=r-\frac{S_h+\log h}{T_h^{\mathrm{unif}}}.
\end{equation*}
This is the potential of the net reproductive function
\(V_h^{\mathrm{unif}}(a):=e^{ra}p_h^{\mathrm{unif}}(a)\), which has
Euler--Lotka rate \(r\). Then
\begin{equation}
\Phi_h^{\mathrm{unif}}-\Phi
=-\bigl(H_h^{\mathrm{unif}}-H\bigr)\longrightarrow0
\qquad(h\to0).
\label{eq:potential_uniform_limit}
\end{equation}
\end{corollary}
The proof is given in Appendix~\ref{app:potential_proofs}.

\par\endgroup

Thus a discrete age-class model with unresolved within-class timing
admits a canonical continuous interpretation: the class probabilities
define a piecewise uniform density, whose cycle entropy is
\(S_h+\log h\), whose generation time is computed from the class
midpoints, and whose evolutionary entropy converges to that of the
underlying continuous process as \(h\to0\).

Exact and approximate reconstruction should be distinguished. When the
true within-class distributions are available, the results of
Sections~\ref{sec:temporal_discretization} and
\ref{sec:change_temporal_mesh} give exact reconstruction and mesh
invariance for every \(h\). When they are unavailable, the uniform
reconstruction provides the corrected entropy \(S_h+\log h\), and the
continuous quantities converge to their true values as \(h\to0\) under the hypotheses of Theorem~\ref{thm:corrected_entropy_continuous_limit}.

\section{The constant-hazard reproductive model}
\label{sec:constant_hazard}

We now consider a demographic model in which the within-class temporal
structure is known exactly. It provides a closed-form realisation of the
four temporal passages in Table~\ref{tab:intro_questions}.

Let \(A>0\) denote age at first reproduction and suppose that
post-maturity mortality is governed by a constant hazard
\(\mu>0\):
\begin{equation*}
\mu(a)=\mu,
\qquad a\geq A.
\end{equation*}
Conditional on survival to age \(A\), post-maturity survivorship is
therefore
\begin{equation*}
\frac{l(a)}{l(A)}
=
e^{-\mu(a-A)},
\qquad a\geq A.
\end{equation*}

Assume that reproduction begins at age \(A\) and that fertility is
constant thereafter:
\begin{equation*}
m(a)
=
\begin{cases}
0, & 0\leq a<A,\\[1mm]
m, & a\geq A.
\end{cases}
\end{equation*}
This implies
\begin{equation*}
V(a)
=
\begin{cases}
0, & 0\leq a<A,\\[1mm]
l(A)m\,e^{-\mu(a-A)}, & a\geq A.
\end{cases}
\end{equation*}

For \(a\geq A\),
\begin{equation*}
p(a)
\, = \, 
e^{-ra}V(a)
\, = \, 
l(A)m\,e^{-rA}
e^{-(\mu+r)(a-A)}.
\end{equation*}
Define
$
\beta=\mu+r,
$
and note that \(\beta>0\). Euler--Lotka normalisation then yields
\[
1
=
l(A)m\,e^{-rA}
\int_A^\infty
e^{-\beta(a-A)}\,da
=
\frac{l(A)m\,e^{-rA}}{\beta},
\]
from which it follows that 
$
l(A)m\,e^{-rA}=\beta.
$
Therefore
\begin{equation}
p_\beta(a)
=
\begin{cases}
0, & 0\leq a<A,\\[1mm]
\beta e^{-\beta(a-A)}, & a\geq A.
\end{cases}
\label{eq:exponential_reproductive_kernel}
\end{equation}

Thus the normalised reproductive-age distribution is a shifted
exponential distribution, and the parameter \(\beta\) is the post-maturity decay rate.

\subsection{Continuous quantities}
\label{sec:continuous_constant_hazard_entropy}

For the exponential reproductive density \eqref{eq:exponential_reproductive_kernel},
\begin{equation*}
T(A,\beta)
=
\int_A^\infty
a\beta e^{-\beta(a-A)}\,da
=
A+\frac{1}{\beta}.
\end{equation*}
With the one-year reference duration understood, the continuous cycle entropy is
\begin{equation}
S(A,\beta)
=
-\int_A^\infty
p_\beta(a)\log p_\beta(a)\,da
=
1-\log\beta.
\label{eq:continuous_exponential_cycle_entropy}
\end{equation}
Hence
\begin{equation}
H(A,\beta)
=
\frac{
1-\log\beta
}{
A+\dfrac{1}{\beta}
}
=
\frac{
\beta(1-\log\beta)
}{
1+A\beta
}.
\label{eq:continuous_exponential_entropy}
\end{equation}

Thus generation time, cycle entropy, and evolutionary entropy all depend only on the physical maturity age $A$ and the continuous decay rate $\beta$.

\subsection{Continuous to continuous: change of temporal unit}
\label{subsec:constant_hazard_continuous_to_continuous}

The continuous-to-continuous passage is an immediate application of Theorem~\ref{pro:reference_scale_change}.
The constant-hazard family is preserved 
under the scaling 
$
\widetilde a=ca,
$
 with
 $\widetilde A=cA,
 \ 
 \widetilde\beta=\frac{\beta}{c}.$
 
Indeed,
\[
\widetilde p(\widetilde a)
=
\widetilde\beta
e^{-\widetilde\beta(\widetilde a-\widetilde A)}
=
\frac{1}{c}
p_\beta\!\left(
\frac{\widetilde a}{c}
\right)
\]
and therefore
\[
\widetilde S=S,
\qquad
\widetilde T=cT,
\qquad
\widetilde H=\frac{H}{c}.
\]
Hence the constant-hazard model provides an explicit realisation of the general change-of-unit theory of Section~\ref{sec:reference_scale}.

\subsection{Continuous to discrete: exact geometric aggregation}
\label{subsec:constant_hazard_continuous_to_discrete}

We aggregate the continuous reproductive-age distribution into
classes of width \(h>0\). Following
Remark~\ref{rem:translated_temporal_grids},  we choose the temporal grid
to begin at the maturity age \(A\). Thus the reproductive classes are
\begin{equation}
I_k^{(h)}
=
[A+kh,A+(k+1)h),
\qquad
k=0,1,2,\ldots.
\label{eq:constant_hazard_age_classes}
\end{equation}

\noindent Over one class interval, post-maturity survival is
$
\rho_h=e^{-\mu h},
$
and the corresponding population growth factor is
$
\lambda_h=e^{rh}.
$
Hence
\begin{equation}
q_h
=
\frac{\rho_h}{\lambda_h}
=
e^{-(\mu+r)h}
=
e^{-\beta h}.
\label{eq:q_h_beta}
\end{equation}
Integrating \eqref{eq:exponential_reproductive_kernel} over the classes
\eqref{eq:constant_hazard_age_classes} gives
\begin{equation*}
P_k^{(h)}
=
\int_{A+kh}^{A+(k+1)h}
\beta e^{-\beta(a-A)}\,da
\nonumber\\
=
(1-e^{-\beta h})e^{-\beta kh}
=
(1-q_h)q_h^k.
\end{equation*}
Thus
$
P_k^{(h)}
=
(1-q_h)q_h^k.
$

The geometric reproductive distribution is therefore not an
approximation to the continuous constant-hazard model, but the exact
class distribution obtained by integration of the shifted
exponential density. The continuous-to-discrete passage induces the parameter transformation 
$
\beta
\longmapsto
q_h=e^{-\beta h}.
$

\subsection{Discrete to continuous: exact reconstruction}
\label{subsec:constant_hazard_discrete_to_continuous}
\label{sec:exact_finite_mesh_reconstruction}

The reverse passage is also exact within the constant-hazard family.

From
\eqref{eq:q_h_beta} it follows that  
$ {\displaystyle \beta
	=
	-\frac{\log q_h}{h}.  }
$
Thus a geometric reproductive tail observed at mesh \(h\), together
with the assumption of a constant-hazard reproductive model, determines
the underlying continuous decay rate exactly. 
Moreover, the within-class reproductive structure is also determined
exactly. Writing 
$
a=A+kh+u$ for 
$
0\leq u<h,
$
the conditional density within the \(k\)-th class is
\begin{equation*}
\varphi_h(u)
=
\frac{
p_\beta(A+kh+u)
}{
P_k^{(h)}
}
\ = \ 
\frac{
\beta e^{-\beta u}
}{
1-e^{-\beta h}
}.
\end{equation*}

A distinctive feature of the exponential model is that
\(\varphi_h\) is independent of \(k\). 
The corresponding within-class cycle entropy is
\begin{equation*}
C_h(\beta)
=
-\int_0^h
\varphi_h(u)\log \varphi_h(u)\,du.
\end{equation*}
Evaluating the integral gives
\begin{equation}
C_h(\beta)
=
1
-
\log\!\left(
\frac{\beta}{1-e^{-\beta h}}
\right)
-
\frac{
\beta h e^{-\beta h}
}{
1-e^{-\beta h}
}.
\label{eq:within_class_entropy_explicit}
\end{equation}

\noindent The ordinary Shannon entropy of the geometric class distribution is
\begin{equation}
S_h
=
-\log(1-q_h)
-
\frac{q_h}{1-q_h}\log q_h.
\label{eq:finite_mesh_class_entropy}
\end{equation}

Theorem~\ref{thm:general_exact_entropy_reconstruction} applies
immediately. Since the within-class entropy is the same in every class, it follows that\newline
$
\sum_{k=0}^\infty
P_k^{(h)}C_h(\beta)
=
C_h(\beta),
$
and therefore
\begin{equation*}
S(A,\beta)
=
S_h+C_h(\beta).
\end{equation*}

Substitution of
\(q_h=e^{-\beta h}\) into
\eqref{eq:finite_mesh_class_entropy} and
\eqref{eq:within_class_entropy_explicit} gives
\[
S_h+C_h(\beta)
=
1-\log\beta,
\]
exactly as in
\eqref{eq:continuous_exponential_cycle_entropy}.
Equivalently, writing 
$\beta=-\frac{\log q}{h},
$
the reconstructed cycle entropy becomes
\begin{equation}
S_h^{\mathrm{exact}}(q)
=
1-
\log\!\left(
\frac{-\log q}{h}
\right).
\label{eq:exact_cycle_entropy_q}
\end{equation}

The same reconstruction applies to generation time.
If the \(k\)-th
class is represented by its left endpoint \(A+kh\), the corresponding
discrete generation time is
\begin{equation*}
T_h^{\mathrm{disc}}(A,q)
=
\sum_{k=0}^\infty
(A+kh)(1-q)q^k
\nonumber\\
=
A+h\frac{q}{1-q}.
\end{equation*}

\noindent The mean reproductive displacement within a class is
\begin{equation*}
\delta_h(\beta)
=
\int_0^h
u \varphi_h(u)\,du
=
\frac{1}{\beta}
-
\frac{
he^{-\beta h}
}{
1-e^{-\beta h}
}.
\end{equation*}
As for the within-class density, the displacement is independent of the class index.
Consequently $T(A,\beta)
=
T_h^{\mathrm{disc}}(A,q_h)
+
\delta_h(\beta)$
or, in terms of \(q\),
\begin{equation*}
\delta_h(q)
=
\frac{h}{-\log q}
-
h\frac{q}{1-q}.
\end{equation*}

\noindent Therefore the exactly reconstructed generation time is
\begin{equation}
T_h^{\mathrm{exact}}(A,q)
=
A+\frac{h}{-\log q}.
\label{eq:exact_generation_time_q}
\end{equation}

The exactly reconstructed evolutionary entropy is thus written in the form

\begin{equation}
	H_h^{\mathrm{exact}}(A,q)
	=
	\frac{
		S_h^{\mathrm{exact}}(q)
	}{
		T_h^{\mathrm{exact}}(A,q)
	} 
=
\frac{
1-
\log\!\left(
\dfrac{-\log q}{h}
\right)
}{
A+\dfrac{h}{-\log q}
}.
\label{eq:exact_finite_mesh_H} 
\end{equation}

\begin{theorem}[Reconstruction from a geometric representation]
\label{thm:exact_temporal_invariance}
Fix \(A>0\), \(\beta>0\) and \(h>0\). 
Then:
\begin{enumerate}
	\item[(i)] $S_h^{\mathrm{exact}}(q_h)
	=
	S(A,\beta),$
	\item[(ii)] $T_h^{\mathrm{exact}}(A,q_h)
	=
	T(A,\beta),$
	\item[(iii)] $H_h^{\mathrm{exact}}(A,q_h)
	=
	H(A,\beta).$
\end{enumerate}
\end{theorem}

\begin{proof}
All the relevant computations were performed above. Substituting $q_h=e^{-\beta h}$ into
\eqref{eq:exact_cycle_entropy_q},
\eqref{eq:exact_generation_time_q}, and
\eqref{eq:exact_finite_mesh_H} yields identities $(i), \  (ii)$ and $(iii)$, respectively.
\end{proof}

Thus, within the constant-hazard family, the discrete-to-continuous
passage is exact: the continuous quantities are recovered from any
finite temporal mesh, without passing to the limit \(h\to0\).

\begin{rem}[Temporal invariance of entropy sensitivity]
In the open-group representation,
$
q_h=e^{-\beta h},
$
so the geometric parameter $q_h$ depends on the temporal mesh, whereas the
physical decay rate \(\beta\) does not. Consequently, the raw sensitivity
$
\partial H_h^{\rm exact}/\partial q_h
$
is mesh-dependent.

Recalling that 
$
H_h^{\rm exact}(A,q_h)=H(A,\beta)$
with $
\beta=-\frac{\log q_h}{h},
$
the chain rule gives
\[
\frac{\partial H_h^{\rm exact}}{\partial q_h}
=
-\frac{1}{hq_h}
\frac{\partial H}{\partial\beta},
\]
or equivalently,
\[
hq_h
\frac{\partial H_h^{\rm exact}}{\partial q_h}
=
-\frac{\partial H}{\partial\beta}.
\]
The scaled sensitivity on the left-hand side therefore depends only on the
physical parameter \(\beta\) and is independent of the temporal mesh.

Thus, if \(h_c\) and \(h_f\) are two temporal meshes describing the
same constant-hazard reproductive process, then
\[
hq_h
\frac{\partial H_h^{\rm exact}}{\partial q_h}
\]
has the same value for both representations.

Since
$ {\displaystyle
\frac{\partial H}{\partial\beta}
=
-\frac{\log\beta+A\beta}
{(1+A\beta)^2}, }
$
we obtain
\[
hq_h
\frac{\partial H_h^{\rm exact}}{\partial q_h}
=
\frac{\log\beta+A\beta}
{(1+A\beta)^2}.
\]

\end{rem}

This reduced sensitivity isolates the effect of the exponential--geometric
parameter, whereas the full matrix sensitivity with respect to the
terminal retention coefficient $\rho$ in open-group Leslie matrices
\cite{Oliveira2026} also incorporates the induced variation of the
Perron growth factor $\lambda$.

\begingroup
\begin{corollary}[Potential in the exponential--geometric family]
\label{cor:potential_exponential}
Fix \(A>0\), \(r\) and the one-year reference, and restrict
\(\beta>0\) to \(\beta>r\), so that \(\mu=\beta-r>0\).
Write \(\Phi(A,\beta;r)\) for the potential of this family.
For a mesh \(h>0\), define
\(\Phi_h^{\mathrm{exact}}(A,q;r):=r-H_h^{\mathrm{exact}}(A,q)\),
where \(0<q<1\) and \(-\log q/h>r\). With \(q_h=e^{-\beta h}\),
\begin{equation}
\Phi_h^{\mathrm{exact}}(A,q_h;r)
=\Phi(A,\beta;r)
=r-\frac{\beta(1-\log\beta)}{1+A\beta}.
\label{eq:potential_exponential_geometric}
\end{equation}
At fixed \(A,r\), its sensitivity satisfies
\begin{equation}
hq_h\frac{\partial\Phi_h^{\mathrm{exact}}}{\partial q}(A,q_h;r)
=-\frac{\partial\Phi}{\partial\beta}(A,\beta;r)
=-\frac{\log\beta+A\beta}{(1+A\beta)^2}.
\label{eq:potential_scaled_sensitivity}
\end{equation}
Thus both the reconstructed potential and this scaled sensitivity are
independent of the mesh for the same physical \(\beta\).
\end{corollary}
The proof is given in Appendix~\ref{app:potential_proofs}.

\par\endgroup

\subsection{Discrete to discrete: exact change of temporal mesh}
\label{subsec:constant_hazard_discrete_to_discrete}
\label{sec:exact_finite_mesh_entropy_rate}

We finally compare two discrete representations of the same
constant-hazard reproductive process. Let \(h_c\) and \(h_f\) denote
the coarse and fine class widths. Since both represent the same
continuous decay rate \(\beta\),
\[
q_{h_c}=e^{-\beta h_c},
\qquad
q_{h_f}=e^{-\beta h_f}.
\]
Eliminating \(\beta\) gives
$
q_{h_f}
=
q_{h_c}^{\,h_f/h_c}
$
or, equivalently,
$ {\displaystyle
-\frac{\log q_{h_c}}{h_c}
=
-\frac{\log q_{h_f}}{h_f}
=
\beta.}
$
Thus, while \(q_h\) is mesh-dependent, 
$
-\frac{\log q}{h}
$
is the corresponding mesh-independent decay rate.

For nested meshes 
$
h_c=mh_f,
$ and 
the splitting coefficients of
Section~\ref{subsec:coarsening_refinement} are obtained explicitly.
Since
$
q_{h_c}=q_{h_f}^{\,m},
$
we have
\begin{equation*}
\alpha_{s}
=
\frac{
(1-q_{h_f})q_{h_f}^{\,s}
}{
1-q_{h_f}^{\,m}
},
\qquad
s=0,\ldots,m-1.
\end{equation*}

These coefficients are independent of the coarse class index. Thus the
general family \(\alpha_{j,s}\) introduced in
Section~\ref{sec:change_temporal_mesh} reduces, in the exponential
case, to a single truncated geometric profile
$
\left(
\alpha_0,\ldots,\alpha_{m-1}
\right)
$
repeated identically in every reproductive class.

Theorem~\ref{thm:exact_temporal_invariance} (iii) specialises to the explicit
identity
\begin{equation*}
H_{h_c}^{\mathrm{exact}}
\left(
A,q_{h_c}
\right)
=
H_{h_f}^{\mathrm{exact}}
\left(
A,q_{h_f}
\right)
=
H(A,\beta),
\end{equation*}
where $
q_{h_f}
=
q_{h_c}^{\,h_f/h_c}.
$

The four passages are therefore completely explicit for the
constant-hazard reproductive model:
\[
\begin{array}{rcl}
C\to C
&:&
(A,\beta)
\mapsto
(cA,\beta/c),
\\[1mm]
C\to D
&:&
\beta
\mapsto
q_h=e^{-\beta h},
\\[1mm]
D\to C
&:&
q_h
\mapsto
\displaystyle
\beta=-\frac{\log q_h}{h},
\\[3mm]
D\to D
&:&
q_{h_c}
\mapsto
q_{h_c}^{\,h_f/h_c}.
\end{array}
\]

All four passages are determined by the single physical parameter
$ {\displaystyle
\beta
=
-\frac{\log q_h}{h}.
}
$

The exact continuous--discrete correspondence may be summarised by the
commutative diagram
\begin{equation*}
\begin{CD}
\beta
@>{\mathcal{E}^{\mathrm{cont}}}>>
\left(S(A,\beta),T(A,\beta),H(A,\beta)\right)
\\
@V{\mathcal{Q}_h:\,\beta\mapsto e^{-\beta h}}VV
@VV{=}V
\\
q_h
@>{\mathcal{E}^{\mathrm{exact}}_h}>>
\left(
S_h^{\mathrm{exact}}(q_h),
T_h^{\mathrm{exact}}(A,q_h),
H_h^{\mathrm{exact}}(A,q_h)
\right).
\end{CD}
\end{equation*}

Here
$
\mathcal{Q}_h(\beta)=e^{-\beta h},
$
and Theorem~\ref{thm:exact_temporal_invariance} shows that the diagram
commutes for every \(h>0\).
Thus temporal aggregation changes only the coordinate used to represent
post-maturity decay,
\[
\beta
\longleftrightarrow
q_h,
\]
while leaving the reconstructed quantities
\(S\), \(T\), and \(H\) unchanged.

\section{Continuous and discrete critical entropy thresholds}
\label{sec:critical_entropy_thresholds}

This section determines the entropy-maximising decay rate for the
constant-hazard reproductive model. Exact reconstruction transports the
continuous optimum to every temporal mesh, yielding a mesh-invariant
critical threshold. We compare it with the optimum obtained from the
conventional discrete age-class formulation after applying the
logarithmic correction to Shannon entropy.
Throughout this section, age and class width are measured in years, with
the one-year reference duration understood.

\subsection{The continuous critical decay rate}
\label{sec:continuous_entropy}

Recall from
\eqref{eq:continuous_exponential_entropy} that
$
H(A,\beta)
=
{\displaystyle \frac{
\beta(1-\log\beta)
}{
1+A\beta}
.}
$
For fixed \(A>0\), differentiation gives
\begin{equation}
\frac{\partial H}{\partial\beta}
=
-
\frac{
\log\beta+A\beta
}{
(1+A\beta)^2
}.
\label{eq:continuous_entropy_derivative}
\end{equation}

\noindent A critical point of $H$ satisfies
$
\log\beta+A\beta=0,
$
or equivalently $\beta= e^{-A\beta}.$

\begin{theorem}[Continuous critical entropy threshold]
\label{thm:continuous_critical_entropy_threshold}

For every \(A>0\), the function
\(\beta\mapsto H(A,\beta)\) has a unique global maximum on
\((0,\infty)\). The maximising decay rate is given by
\begin{equation}
\beta_c(A)
=
\frac{W(A)}{A},
\label{eq:continuous_critical_beta}
\end{equation}
where \(W\) denotes the principal Lambert \(W\)-function.
\end{theorem}

\begin{proof}
The function
$g(\beta)=\log\beta+A\beta$
is strictly increasing because
$
g'(\beta)=\frac{1}{\beta}+A>0.
$
Moreover,
$
g(\beta)\longrightarrow-\infty$
as $\beta\to0^+$
and
$
g(\beta)\longrightarrow+\infty
 \ 
\text{as }\beta\to\infty,
$
and so $\log\beta+A\beta=0$ has a unique
solution $\beta_c$. By \eqref{eq:continuous_entropy_derivative},
\(H\) is increasing for $\beta < \beta_c$ and decreasing for $\beta > \beta_c$. Thus $\beta_c$ is the unique global maximum.

Since
$
A\beta_c e^{A\beta_c}=A,
$
we obtain
$
A\beta_c=W(A),
$
proving \eqref{eq:continuous_critical_beta}.

Finally, note that  $
\log\beta_c=-A\beta_c
$, and thus 
$
1-\log\beta_c
=
1+A\beta_c.
$
Substitution into
\eqref{eq:continuous_exponential_entropy} then gives
\[
H(A,\beta_c)
=
\frac{
\beta_c(1+A\beta_c)
}{
1+A\beta_c
}
\]
which proves
$
H(A,\beta_c) =
\beta_c.
$
\end{proof}

Thus the maximum evolutionary entropy equals the critical decay rate:
\[
H_{\max}=\beta_c.
\]
By exact reconstruction, the same critical decay rate is obtained at
every temporal resolution.
For a general reference \(\tau>0\), the critical equation is \(\log(\tau\beta_c)+A\beta_c=0\), giving \(\beta_c(A;\tau)=W(A/\tau)/A\). Thus mesh invariance holds for each fixed reference duration.

\subsection{Exact transport of the critical point to a temporal mesh}
\label{sec:exact_mesh_invariant_threshold}

For a mesh of width \(h>0\), the continuous decay rate \(\beta\) is represented by
\[
q_h=e^{-\beta h}.
\]
Hence the continuous critical decay rate \(\beta_c(A)\) corresponds to
\begin{equation}
q_{c,h}^{\mathrm{exact}}
=
e^{-\beta_c(A)h}.
\label{eq:exact_critical_q}
\end{equation}
By Theorem~\ref{thm:exact_temporal_invariance},
\begin{equation*}
H_h^{\mathrm{exact}}
\left(A,q_{c,h}^{\mathrm{exact}}\right)
=
H\!\left(A,\beta_c(A)\right)
=
\beta_c(A).
\end{equation*}

Substituting
\[
\beta_c
=
-\frac{\log q_{c,h}^{\mathrm{exact}}}{h}
\]
into \(\beta_c=e^{-A\beta_c}\) gives
\begin{equation*}
-\log q_{c,h}^{\mathrm{exact}}
=
h
\left(q_{c,h}^{\mathrm{exact}}\right)^{A/h}.
\end{equation*}
Although \(q_{c,h}^{\mathrm{exact}}\) depends on the temporal mesh, the corresponding decay rate does not:
\begin{equation*}
-\frac{\log q_{c,h}^{\mathrm{exact}}}{h}
=
\beta_c(A)
=
\frac{W(A)}{A}.
\end{equation*}
Thus the exact geometric critical point is the mesh-dependent representation of the unique continuous critical decay rate.

\begingroup
\begin{corollary}[Exact transport of the potential minimum]
\label{cor:potential_minimum}
Fix \(A>0\), the one-year reference and a growth rate
\(r<\beta_c(A)\). On the admissible domain
\(\beta>\max\{0,r\}\), the function \(\Phi(A,\beta;r)\) has its unique
minimum at \(\beta_c(A)\). For every \(h>0\), the corresponding minimum
of \(\Phi_h^{\mathrm{exact}}\) is attained at the point
\(q_{c,h}^{\mathrm{exact}}\) defined in \eqref{eq:exact_critical_q}, and
\begin{equation*}
\Phi_h^{\mathrm{exact}}(A,q_{c,h}^{\mathrm{exact}};r)
=\Phi(A,\beta_c(A);r)=r-\beta_c(A).
\end{equation*}
\end{corollary}
The proof is given in Appendix~\ref{app:potential_proofs}.

The restriction to fixed \(r\) is essential: a change in \(r\) contributes
separately to the variation of \(\Phi=r-H\).
\par\endgroup

\subsection{Comparison with the conventional discrete threshold}
\label{sec:renormalized_discrete_threshold}

For the geometric reproductive distribution
\[
P_k=(1-q)q^k,
\qquad 0<q<1,
\]
the Shannon entropy is
\[
S_h(q)
=
-\log(1-q)
-
\frac{q}{1-q}\log q.
\]
Section~\ref{sec:unresolved_within_class_structure} shows that the logarithmically corrected cycle entropy for a mesh of width \(h\) is \(S_h(q)+\log h\).

To retain the discrete age-class convention used in
\cite{BuescuElaydiOliveira2026}, we represent the classes by their left endpoints, giving
\begin{equation*}
T_h^{\mathrm{disc}}(A,q)
=
A+h\frac{q}{1-q}.
\end{equation*}
The corrected discrete entropy rate is therefore the quotient
\[
\frac{S_h(q)+\log h}
{A+h\dfrac{q}{1-q}}.
\]
This left-endpoint convention differs from \(H_h^{\mathrm{unif}}\), whose denominator uses class midpoints and is \(T_h^{\mathrm{disc}}+h/2\).
The derivative of the corrected discrete quotient with left-endpoint generation time is
\begin{equation}
\frac{
h
}{
(1-q)^2
\left(
A+h\dfrac{q}{1-q}
\right)^2
}
\log\!\left[
\frac{1-q}{h q^{A/h}}
\right].
\label{eq:renormalized_discrete_H_derivative}
\end{equation}
Hence the critical point is determined by
\begin{equation*}
1-q
=
h q^{A/h}.
\end{equation*}

\begin{theorem}[Corrected discrete critical threshold]
\label{thm:renormalized_discrete_threshold}

For every \(A>0\) and \(h>0\), the quotient
\[
\frac{S_h(q)+\log h}
{A+h\dfrac{q}{1-q}}
\]
has a unique maximum on \(0<q<1\). Its maximising value
\(q_{c,h}^{\mathrm{disc}}\) is the unique solution of
\begin{equation}
1-q_{c,h}^{\mathrm{disc}}
=
h
\left(q_{c,h}^{\mathrm{disc}}\right)^{A/h}.
\label{eq:discrete_critical_q}
\end{equation}
At the maximum,
\begin{equation}
\frac{
S_h\!\left(q_{c,h}^{\mathrm{disc}}\right)+\log h
}{
A+h\dfrac{q_{c,h}^{\mathrm{disc}}}
{1-q_{c,h}^{\mathrm{disc}}}
}
=
-\frac{\log q_{c,h}^{\mathrm{disc}}}{h}.
\label{eq:discrete_critical_entropy_value}
\end{equation}
\end{theorem}

\begin{proof}
The function
\[
q\longmapsto\frac{1-q}{h q^{A/h}}
\]
is strictly decreasing from \(+\infty\) to \(0\) on \((0,1)\), so
\eqref{eq:discrete_critical_q} has a unique solution. Equation
\eqref{eq:renormalized_discrete_H_derivative} shows that the derivative is positive before this solution and negative after it.

At the critical point,
\[
\log(1-q)
=
\log h+\frac{A}{h}\log q,
\]
and therefore
\[
S_h(q)+\log h
=
-\log q
\left(
\frac{A}{h}
+
\frac{q}{1-q}
\right).
\]
Since
\[
A+h\frac{q}{1-q}
=
h
\left(
\frac{A}{h}
+
\frac{q}{1-q}
\right),
\]
equation \eqref{eq:discrete_critical_entropy_value} follows.
\end{proof}

At their respective maxima, both entropy rates equal their corresponding decay rate \(\beta\), but the maximising values differ. The exact threshold uses the reconstructed continuous model, whereas the conventional threshold retains the left-endpoint generation time and uses \(\log h\), rather than \(C_h(\beta)\), for within-class entropy.

Define
\begin{equation*}
\beta_{c,h}^{\mathrm{disc}}
=
-\frac{\log q_{c,h}^{\mathrm{disc}}}{h}.
\end{equation*}
Substituting
\(q_{c,h}^{\mathrm{disc}}
=e^{-\beta_{c,h}^{\mathrm{disc}}h}\)
into \eqref{eq:discrete_critical_q} gives
\begin{equation*}
\frac{
1-e^{-\beta_{c,h}^{\mathrm{disc}}h}
}{h}
=
e^{-A\beta_{c,h}^{\mathrm{disc}}}.
\end{equation*}
The continuous critical equation is
\[
\beta_c=e^{-A\beta_c}.
\]
Thus the two equations differ through
\[
\beta
\qquad\longleftrightarrow\qquad
\frac{1-e^{-\beta h}}{h}.
\]

Equivalently, in geometric coordinates,
\begin{align*}
-\log q
&=
h q^{A/h},
\\
1-q
&=
h q^{A/h}.
\end{align*}
Since \(1-q<-\log q\) for \(0<q<1\),
\begin{equation*}
q_{c,h}^{\mathrm{disc}}
<
q_{c,h}^{\mathrm{exact}},
\end{equation*}
and hence
\begin{equation*}
\beta_{c,h}^{\mathrm{disc}}
>
\beta_c.
\end{equation*}

\begin{theorem}[Continuous limit of the corrected discrete threshold]
\label{thm:discrete_critical_threshold_limit}

For fixed \(A>0\),
\begin{equation*}
\lim_{h\to0}
\beta_{c,h}^{\mathrm{disc}}
=
\beta_c(A)
=
\frac{W(A)}{A}.
\end{equation*}
More precisely,
\begin{equation}
\beta_{c,h}^{\mathrm{disc}}
=
\beta_c
+
\frac{
\beta_c^2
}{
2(1+A\beta_c)
}
h
+
O(h^2).
\label{eq:discrete_critical_beta_asymptotic}
\end{equation}
\end{theorem}

\begin{proof}
Let
\[
F(\beta,h)
=
\frac{1-e^{-\beta h}}{h}
-
e^{-A\beta},
\qquad
F(\beta,0)
=
\beta-e^{-A\beta}.
\]
At \((\beta_c,0)\),
\[
F(\beta_c,0)=0,
\qquad
\frac{\partial F}{\partial\beta}(\beta_c,0)
=
1+A\beta_c>0.
\]
The implicit function theorem gives
\(\beta_{c,h}^{\mathrm{disc}}\to\beta_c\) as \(h\to0\).
Using
\[
\frac{1-e^{-\beta h}}{h}
=
\beta-\frac{\beta^2}{2}h+O(h^2)
\]
and writing
\[
\beta_{c,h}^{\mathrm{disc}}
=
\beta_c+ch+O(h^2),
\]
we obtain
\[
c(1+A\beta_c)-\frac{\beta_c^2}{2}=0,
\]
which proves \eqref{eq:discrete_critical_beta_asymptotic}.
\end{proof}

Since
$
1-q\sim-\log q$
as $q\to1,
$
the two threshold equations become asymptotically equivalent as the mesh is refined. They therefore give nearly identical critical decay rates for fine temporal resolutions and, in particular, for long-lived populations. Their agreement reflects the approximation
\[
1-q\approx-\log q
\]
rather than exact temporal invariance.

\section{Greenland shark across temporal representations}
\label{sec:greenland_shark_temporal_correspondence}

The Greenland shark
(\emph{Somniosus microcephalus}) \cite{BuescuElaydiOliveira2026GreenlandShark}
provides a numerical illustration of the temporal-invariance results.
Following the demographic estimates of~\cite{Nielsen2016GreenlandShark}, 
we
take
$
A=156\ \mathrm{years}
$
and consider the three longevity-calibrated annual geometric scenarios of~\cite{BuescuElaydiOliveira2026GreenlandShark}
\begin{equation*}
q_1\simeq 0.974613,\quad 0.987413,\quad 0.991632.
\end{equation*}
These are the midpoints of the intervals for which the annual 95\% reproductive quantile is 272, 392, or 512 years, respectively. They are model-based reproductive-persistence ratios \(q_1=\rho_1/\lambda_1\), not direct estimates of adult survival; all derived values below are computed before rounding. For the annual mesh,
\[
S_1
=
-\log(1-q_1)
-
\frac{q_1}{1-q_1}\log q_1,
\qquad
T_1^{\mathrm{disc}}
=
A+\frac{q_1}{1-q_1}.
\]
The corresponding continuous quantities are obtained exactly from
\(\beta=-\log q_1\) using the reconstruction formulas derived in
Section~\ref{sec:constant_hazard}. The quantities \(S,T,H\) below refer to this continuous reconstruction, whereas the analysis in~\cite{BuescuElaydiOliveira2026GreenlandShark} uses the annual discrete convention.

We now represent the same demographic process using classes of width
\(h=20\) years, anchored at the maturity age \(A\). The corresponding
geometric parameter is
\begin{equation*}
q_{20}
=
e^{-20\beta}
=
q_1^{20}.
\end{equation*}
The raw discrete quantities are
\[
S_{20}
=
-\log(1-q_{20})
-
\frac{q_{20}}{1-q_{20}}\log q_{20},
\qquad
T_{20}^{\mathrm{disc}}
=
A+20\frac{q_{20}}{1-q_{20}},
\]
whereas exact reconstruction gives
\[
S_{20}^{\mathrm{rec}}
=
S_{20}+C_{20}(\beta)
=
S,
\qquad
T_{20}^{\mathrm{rec}}
=
T_{20}^{\mathrm{disc}}+\delta_{20}(\beta)
=
T.
\]

\begin{table}[ht]
\centering
\scriptsize
\caption{Annual and twenty-year representations of the three longevity-calibrated Greenland
shark scenarios of~\cite{BuescuElaydiOliveira2026GreenlandShark}.}
\vspace{2mm}
\label{tab:greenland_temporal_comparison}
\renewcommand{\arraystretch}{1.12}
\setlength{\tabcolsep}{2.5pt}
\begin{tabular}{ccccccccc}
\toprule
\(q_1\)
&
\(q_{20}\)
&
\(S_1\)
&
\(S_{20}\)
&
\(S\)
&
\(T_1^{\rm disc}\)
&
\(T_{20}^{\rm disc}\)
&
\(T\)
&
\(H\)
\\
\midrule
0.974613
& 0.597920
& 4.660712
& 1.675901
& 4.660684
& 194.390
& 185.741
& 194.888
& 0.023915 \\
0.987413
& 0.776204
& 5.368762
& 2.375693
& 5.368755
& 234.446
& 225.367
& 234.945
& 0.022851 \\
0.991632
& 0.845300
& 5.779147
& 2.784587
& 5.779144
& 274.503
& 265.282
& 275.002
& 0.021015 \\
\bottomrule
\end{tabular}
\end{table}

Table~\ref{tab:greenland_temporal_comparison} illustrates the effect of
changing the temporal mesh. The ordinary Shannon entropy and the
conventional left-endpoint generation time vary substantially between
the annual and twenty-year representations, whereas exact
reconstruction recovers the same cycle entropy, physical generation
time, and evolutionary entropy in every case.

The uniform reconstruction gives
\[
S_{20}^{\mathrm{unif}}
=
S_{20}+\log 20,
\qquad
S_{20}^{\mathrm{unif}}-S
=
D_{\mathrm{KL}}
\left(
p\,\middle\|\,p_{20}^{\mathrm{unif}}
\right).
\]
For the three scenarios, the corresponding Kullback--Leibler divergences are
\[
0.010949,\quad 0.002670,\quad 0.001176.
\]
Even with twenty-year age classes, the piecewise uniform
reconstruction therefore provides an excellent approximation to the
continuous entropy $S$, while exact reconstruction recovers it
identically. 

For \(A=156\), the continuous critical decay rate is
\begin{equation*}
\beta_c
=
\frac{W(156)}{156}
=
0.023927686.
\end{equation*}
The exact critical parameter satisfies
\[
q_{c,h}^{\mathrm{exact}}
=
e^{-\beta_c h},
\]
whereas the corrected conventional discrete critical parameter satisfies
\[
1-q_{c,h}^{\mathrm{disc}}
=
h
\left(
q_{c,h}^{\mathrm{disc}}
\right)^{A/h}.
\]

\begin{table}[ht]
\centering
\small
\caption{Exact and corrected conventional discrete critical thresholds.}\medskip
\label{tab:greenland_critical_threshold_meshes}
\renewcommand{\arraystretch}{1.12}
\begin{tabular}{ccccc}
\toprule
\(h\)
&
\(q_{c,h}^{\rm exact}\)
&
\(q_{c,h}^{\rm disc}\)
&
\(\beta_c\)
&
\(\beta_{c,h}^{\rm disc}\)
\\
\midrule
1
& 0.976356311
& 0.976297209
& 0.023927686
& 0.023988221
\\
20
& 0.619678969
& 0.604682847
& 0.023927686
& 0.025152559
\\
\bottomrule
\end{tabular}
\end{table}

For \(h=1\), \(q_{c,1}^{\mathrm{disc}}=0.976297209\) 
	is the critical persistence threshold used in~\cite{BuescuElaydiOliveira2026GreenlandShark}.
Table~\ref{tab:greenland_critical_threshold_meshes} illustrates the
distinction between the exact and corrected discrete thresholds. The
exact threshold preserves the continuous critical decay rate
\[
-\frac{\log q_{c,h}^{\mathrm{exact}}}{h}
=
\beta_c
\]
at every temporal resolution, whereas the corrected conventional
threshold replaces the continuous decay rate by its finite-difference
approximation \((1-q)/h\). As predicted by
Theorem~\ref{thm:discrete_critical_threshold_limit}, the discrepancy
remains small at annual resolution and increases for coarser meshes.

Thus the geometric parameter, the ordinary Shannon entropy, the
conventional generation time, and the corrected conventional discrete
critical threshold all depend on the temporal mesh. Exact reconstruction,
however, recovers the same cycle entropy, physical generation time,
evolutionary entropy, and critical decay rate independently of the
chosen temporal representation.

\begingroup
For these Greenland shark scenarios, \(A\) and \(q_h\) determine
\(\beta\), but not \(r\) separately. Numerical values of \(\Phi\) therefore
require additional growth or survival information
\cite{BuescuElaydiOliveira2026GreenlandShark}.
\par\endgroup

\section{Conclusions}
\label{sec:conclusions}

Age-structured demographic models require a choice of temporal mesh, but this modelling decision has no counterpart in the underlying biology. 
The principal result is that complete temporal representations are equivalent, and that evolutionary entropy, at a fixed reference duration and with the necessary within-class information retained, is a mesh-invariant property of the reproductive process.
For the same demographic process and a fixed reference duration,
Corollary~\ref{cor:potential_mesh_invariance} extends this invariance to
reproductive potential through \(r=H+\Phi\).

Three conclusions follow from this viewpoint.
Firstly, the temporal mesh labels a representation; aggregation alone is not an invertible change of coordinates.
The four passages between continuous and discrete representations summarised in Table~\ref{tab:intro_questions} form a coherent structure in which the physical reproductive-age distribution is the invariant object and the mesh width $h$ serves as a coordinate label. Raw discrete quantities~-- class probabilities, Shannon entropy, left-endpoint generation time, geometric decay parameters~-- transform non-trivially under a change of mesh and therefore carry mixed information about the reproductive process and its temporal partition. Reconstructed cycle entropy $S_h^{\mathrm{rec}}$, physical generation time $T$, and evolutionary entropy $H$ do not. In this precise sense, and for a fixed reference duration, evolutionary entropy is intrinsic to the reproductive process, while its apparent dependence on age-class width is an artefact of unresolved within-class timing.

Secondly, the error after logarithmic correction is a Kullback--Leibler divergence. 
When within-class timing is unavailable, the piecewise uniform
reconstruction identifies \(S_h+\log h\) as the cycle entropy of a
well-defined continuous density. The identity
\begin{equation*}
	S_h+\log h-S
	=
	D_{\mathrm{KL}}
	\!\left(
	p\,\big\|\,p_h^{\mathrm{unif}}
	\right)
\end{equation*}
shows that the corrected entropy error measures exactly the divergence of the true reproductive-age density from its piecewise uniform reconstruction. It is non-negative, vanishes if and only if \(p\) is constant almost everywhere on each non-empty class, and tends to zero under the hypotheses of Theorem~\ref{thm:corrected_entropy_continuous_limit}. Exact reconstruction eliminates this error at every mesh; uniform reconstruction does so asymptotically under those hypotheses.

Finally, constant post-maturity hazard is the exactly solvable instance of the theory.
For constant hazard and fertility beyond the age at first reproduction, the four passages between continuous and discrete representations reduce to the single parameter transformation
\begin{equation*}
	\beta \;\longleftrightarrow\; q_h = e^{-\beta h},
\end{equation*}
and Theorem~\ref{thm:exact_temporal_invariance} shows that continuous and discrete descriptions coincide at every mesh, not merely in the limit.
It is an exactly solvable model in which the general theory becomes completely explicit.
It also clarifies the relation between the continuous critical decay rate $\beta_c(A) = W(A)/A$ and the corrected discrete threshold: the two coincide only through the first-order approximation $1 - q \approx -\log q$, which explains their near-agreement for long-lived populations without conflating them.

\medskip

Our results also have practical implications for comparative demography.
The results apply directly to age-class models satisfying the stated reconstruction assumptions, not automatically to arbitrary stage-classified matrices. Comparisons require a common physical time unit and reference duration, together with the necessary within-class information or an explicitly identified approximation; a matrix projection interval is not itself an age-class width. The same considerations apply to
entropy-maximising thresholds and, more generally, to comparative
statements about reproductive strategy. What is invariant across
species is not necessarily invariant across temporal meshes, and
conversely.

For stage-classified matrices, Salguero-G{\'o}mez and Plotkin
\cite{SalgueroGomezPlotkin2010} showed that aggregation criteria can bias
demographic elasticities. This is a distinct aggregation problem, rather than
a change of age-class width, and reinforces the need to specify what a
comparison preserves.

Evolutionary entropy is therefore, for a fixed reference duration, an intrinsic property of the biological reproductive process itself rather than of any particular complete age-class representation.

In a subsequent paper, we will apply these results to compare evolutionary
entropy and reproductive potential across species represented by demographic
models with different projection intervals.

\begingroup
\section*{Statements and Declarations}

\subsection*{Funding}
Jorge Buescu acknowledges partial support from Funda\c{c}\~ao para a Ci\^encia e a
Tecnologia through project UID/04561/2025
(\url{https://doi.org/10.54499/UID/04561/2025}).
Henrique M. Oliveira acknowledges partial support from Funda\c{c}\~ao para a
Ci\^encia e a Tecnologia through CAMGSD---Centro de An\'alise Matem\'atica,
Geometria e Sistemas Din\^amicos, under projects UID/04459/2025
(\url{https://doi.org/10.54499/UID/04459/2025}) and UID/PRR/04459/2025
(\url{https://doi.org/10.54499/UID/PRR/04459/2025}).

\subsection*{Competing interests}
The authors declare no competing interests.

\subsection*{Author contributions}
Both authors contributed to the conception of the study, the mathematical
analysis, and the preparation of the manuscript.

\subsection*{Data and code availability}
No new empirical data were collected. The demographic inputs for the Greenland
shark illustration are taken from
\cite{Nielsen2016GreenlandShark,BuescuElaydiOliveira2026GreenlandShark}.
All model inputs and formulae needed to reproduce the numerical values are
provided in the manuscript and the cited calibration study. The calculations require only direct evaluation
of these formulae and one-dimensional root finding; no dedicated software
package is required.

\subsection*{Ethics approval and consent to participate}
Not applicable. This theoretical study involved no human participants and no
animal experiments.

\subsection*{Consent for publication}
Not applicable. No individual participant data are reported.

\subsection*{Use of artificial intelligence}
ChatGPT (OpenAI) was used to assist with language editing, bibliographic searches and 
LaTeX preparation. The authors are fully responsible for the  scientific content.
\par\endgroup

\begingroup
\appendix
\section{Proofs of results for the reproductive potential}
\label{app:potential_proofs}
\label{sec:reproductive_potential_invariance}

\begin{proof}[Proof of Proposition~\ref{pro:potential_decomposition}]
From \(p(a)=e^{-ra}V(a)\),
\(\log(\tau V(a))=ra+\log(\tau p(a))\) wherever \(p(a)>0\).
Multiplication by \(p(a)\) and integration give
\begin{equation*}
T\Phi_\tau=r\int_0^\infty a p(a)\,da
+\int_0^\infty p(a)\log(\tau p(a))\,da
=rT-S_\tau.
\end{equation*}
Division by \(T>0\) proves the identity.
\end{proof}

\begin{proof}[Proof of Corollary~\ref{cor:potential_unit_reference}]
Theorem~\ref{pro:reference_scale_change} gives
\(\widetilde H_{\widetilde\tau}=H_\tau/c\), and the same change of unit gives
\(\widetilde r=r/c\). Subtracting in
\eqref{eq:entropy_potential_decomposition} proves
\(\widetilde\Phi_{\widetilde\tau}=\Phi_\tau/c\).
For a change of reference at fixed age coordinate, use
\(\log(\tau_2p)=\log(\tau_1p)+\log(\tau_2/\tau_1)\)
in \eqref{eq:continuous_cycle_entropy} and \(\int p=1\).
This gives the stated change in \(S_\tau\), and division by \(T\) gives
the change in \(H_\tau\). Euler--Lotka does not involve \(\tau\), so \(r\)
is unchanged; the change in \(\Phi_\tau\) follows by subtraction.
\end{proof}

\begin{proof}[Proof of Corollary~\ref{cor:potential_exact_reconstruction}]
Theorem~\ref{thm:general_exact_entropy_reconstruction} identifies the
numerator in \eqref{eq:potential_exact_reconstruction} with \(S\).
Equation~\eqref{eq:generation_time_class_units} and
Proposition~\ref{pro:general_exact_generation_reconstruction} identify its
denominator with \(T\). The reconstructed potential is therefore
\(r-S/T=r-H=\Phi\).
\end{proof}

\begin{proof}[Proof of Corollary~\ref{cor:potential_mesh_invariance}]
Theorem~\ref{thm:evolutionary_entropy_mesh_invariance} states that both
reconstructed entropy rates equal \(H\). Since both representations retain
the same \(r\), their potentials equal \(r-H=\Phi\), by
\eqref{eq:potential_exact_reconstruction}.
\end{proof}

\begin{proof}[Proof of Corollary~\ref{cor:potential_uniform_limit}]
The reconstructed density has integral one. Hence
\(\int_0^\infty e^{-ra}V_h^{\mathrm{unif}}(a)\,da=1\), and applying
Proposition~\ref{pro:potential_decomposition} to this density gives
\(\Phi_h^{\mathrm{unif}}=r-H_h^{\mathrm{unif}}\).
Corollary~\ref{cor:corrected_evolutionary_entropy_limit} gives
\(H_h^{\mathrm{unif}}\to H\). Subtracting from the fixed \(r\) proves
\eqref{eq:potential_uniform_limit}.
\end{proof}

\begin{proof}[Proof of Corollary~\ref{cor:potential_exponential}]
Subtract the entropy formula \eqref{eq:continuous_exponential_entropy}
from \(r\). Theorem~\ref{thm:exact_temporal_invariance} gives
\(H_h^{\mathrm{exact}}(A,q_h)=H(A,\beta)\), proving
\eqref{eq:potential_exponential_geometric}.
At fixed \(A,r\), differentiating \(\Phi=r-H\) gives
\(\partial\Phi/\partial\beta=-\partial H/\partial\beta\).
Since \(\beta=-\log q/h\), the chain rule gives
\(hq\,\partial\Phi_h^{\mathrm{exact}}/\partial q
=-\partial\Phi/\partial\beta\).
The expression \eqref{eq:continuous_entropy_derivative} for \(\partial H/\partial\beta\)
then proves \eqref{eq:potential_scaled_sensitivity}.
\end{proof}

\begin{proof}[Proof of Corollary~\ref{cor:potential_minimum}]
At fixed \(r\), minimising \(\Phi=r-H\) is equivalent to maximising \(H\).
Theorem~\ref{thm:continuous_critical_entropy_threshold} gives the unique
maximum at \(\beta_c(A)\), where \(H(A,\beta_c(A))=\beta_c(A)\).
The assumption \(r<\beta_c(A)\) places this point in the admissible domain.
The map \(\beta\mapsto e^{-\beta h}\) is one-to-one, and
Corollary~\ref{cor:potential_exponential} preserves the values of the
potential under this map. This proves the claimed minimum and its exact
transport to each mesh.
\end{proof}
\par\endgroup

\begingroup
\bibliographystyle{plain}
\bibliography{BibloH2}
\endgroup
\end{document}